\documentclass[draft,english,twoside,11pt]{article}

\usepackage[lmargin=1in,rmargin=1in,tmargin=1in,bmargin=1in]{geometry}

\usepackage{float}
\usepackage{graphicx}
\usepackage{color}
\usepackage{tikz}
\usetikzlibrary{backgrounds}

\usepackage[margin=20pt]{caption}

\usepackage{amsmath}
\usepackage{amssymb}
\usepackage{amsthm}
\usepackage{mathrsfs}

\usepackage[english]{babel}
\usepackage{flafter}
\usepackage{array}
\usepackage{paralist}
\usepackage{todonotes}
\usepackage{cleveref}

\usepackage{authblk}

\usepackage{xspace}
\usepackage{mathtools}
\usepackage{complexity}

\graphicspath{{.}{graphics/}}

\newtheorem{theorem}{Theorem}
\newtheorem{corollary}[theorem]{Corollary}
\newtheorem*{corollary*}{Corollary}
\newtheorem{lemma}[theorem]{Lemma}
\newtheorem{proposition}[theorem]{Proposition}
\newtheorem*{proposition*}{Proposition}

\newtheorem*{property*}{Property}

\newtheorem{myclaim}{Claim}

\theoremstyle{definition}
\newtheorem{definition}[theorem]{Definition}
\newtheorem*{definition*}{Definition}

\theoremstyle{remark}

\newtheorem*{remark*}{Remark}

\newtheorem*{example*}{Example}

\newcommand{\Z}{\mathbb Z}    

\providecommand{\dsnp}{Directed Steiner Forest\xspace}
\newcommand*{\twokdst}{2-connected Directed $k$ Steiner Tree\xspace}
\newcommand*{\onetwotwodst}{1-2-connected Directed $2$ Steiner Tree\xspace}

\newcommand*{\dstp}{Directed Steiner Tree\xspace}

\newcommand*{\mcfp}{MCF\xspace}

\newcommand*{\onerkdpa}{{Fault-Tolerant $\ell$-Flow Augmentation}\xspace}
\newcommand*{\dortf}{{Fault Tolerant $2$-Flow}\xspace}
\newcommand*{\dorkf}{{Fault-Tolerant $\ell$-Flow}\xspace}
\newcommand*{\dorkkf}{{Fault-Tolerant $k$-Flow}\xspace}
\newcommand*{\ftflong}{{Fault-Tolerant Flow}\xspace}
\newcommand*{\ftf}{{FTF}\xspace}
\newcommand*{\ftp}{{FTP}\xspace}

\newcommand*{\wrma}{Weighted Robust Matching Augmentation\xspace}

\usepackage[format=hang,justification=justified,subrefformat=parens,labelformat=parens]{subcaption}

\newcommand*{\opt}{\ensuremath{\operatorname{OPT}}\xspace}

\tikzset{
	edge/.style={very thick, gray},
	medge/.style={decorate,very thick,decoration={snake}},
	aedge/.style={very thick,dashed,black},
	dedge/.style={thick,->},
	availedge/.style={thick,blue},
	vertex/.style={shape=circle,thick,draw,node distance=3em}
}

\usepackage{framed}
\usepackage{algorithm}
\usepackage[noend]{algorithmic}

\input{mathdots}

\title{Fault-Tolerant Edge-Disjoint Paths --- Beyond Uniform Faults}

\author[1]{David Adjiashvili}
\author[2]{Felix Hommelsheim\thanks{Research partially supported by the German Research Foundation (DFG), RTG 1855}}
\author[3]{Moritz M\"uhlenthaler}
\author[4]{Oliver Schaudt}
\affil[1]{Department of Mathematics, ETH Z\"urich}
\affil[2]{Department of Mathematics, TU Dortmund University}
\affil[2]{Laboratoire G-SCOP, Grenoble INP, Univ. Grenoble-Alpes}
\affil[4]{Department of Mathematics, RWTH Aachen University}

\begin{document}

\maketitle

\begin{abstract}
The overwhelming majority of survivable (fault-tolerant) network design models assume a uniform fault model. Such a model assumes that every subset of the network resources (edges or vertices) of a given cardinality $k$ may fail. While this approach yields problems with clean combinatorial structure and good algorithms, it often fails to capture the true nature of the scenario set coming from applications.  One natural refinement of the uniform model is obtained by partitioning the set of resources into \textit{vulnerable} and \textit{safe} resources. The scenario set contains every subset of at most $k$ faulty resources. This work studies the \textit{Fault-Tolerant Path} (FTP) problem, the counterpart of the Shortest Path problem in this fault model and the \textit{Fault-Tolerant Flow} problem (\ftf), the counterpart of the $\ell$-disjoint Shortest $s$-$t$ Path problem. We present complexity results alongside exact and approximation algorithms for both models. We emphasize the vast increase in the complexity of the problem with respect to the uniform analogue, the Edge-Disjoint Paths problem.
\end{abstract}

\newpage

\section{Introduction}\label{sec:intro}

The \emph{Minimum-Cost Edge-Disjoint Path} (EDP) problem is a classical
network design problem, defined as follows.  Given an edge-weighted directed graph
$D=(V,A)$, two terminals $s,t\in V$ and an integer parameter $k\in
\mathbb{Z}_{\geq 0}$, find $k$ edge-disjoint paths connecting $s$ and $t$ with minimum
total cost. EDP is motivated by the following survivable network design
problem: what is the connection cost of two nodes in a network, given that any
$k-1$ edges can be a-posteriori removed from the graph. The implicit assumption
in EDP is that every edge in the graph in equally vulnerable. Unfortunately,
this assumption is unrealistic in many applications, hence resulting in
overly-conservative solutions.  Our goal is to advance the understanding of
non-uniform models for network design problems in order to avoid solutions that
are too conservative and hence, too costly.  To this end we study a natural
generalization of the EDP problem called the \emph{Fault-Tolerant Path} (FTP)
problem, 
in which we consider a subset of the edges to be vulnerable.  The problem asks
for a minimum-cost subgraph of a given graph that contains an $s$-$t$ path
after removing any $k$ vulnerable edges from the graph. Formally, it is defined
as follows.

\begin{framed}
	\begin{tabular}[h]{lp{0.8\linewidth}}
		\multicolumn{2}{l}{\textit{Fault-Tolerant Path (FTP)}} \tabularnewline
		\textbf{Instance:} &  edge-weighted directed graph $D=(V,A)$, two nodes $s,t\in V$, subset $M\subseteq A$ of vulnerable edges, and integer
		$k\in \mathbb{Z}_{\geq 0}$.\tabularnewline
		\textbf{Task:}    & Find minimum-cost set $S \subseteq A$, such that $S \setminus F$ contains an $s$-$t$ path for every $F\subseteq M$ with
		$|F| \leq k$.
	\end{tabular}
\end{framed}

Observe that \ftp becomes EDP when $M = A$. We will also study EDP with a
simpler, but still non-uniform, fault model: The problem \ftflong (\ftf) asks
for $\ell \geq 1$ fault-tolerant disjoint $s$-$t$ paths, assuming that only a
single edge can be a-posteriori removed from the graph. The problem is defined
as follows.
\begin{framed}
	\begin{tabular}[h]{lp{0.8\linewidth}}
		\multicolumn{2}{l}{\textit{\ftflong (\ftf)}} \tabularnewline
		\textbf{Instance:} &  edge-weighted directed graph $D=(V, A)$, two nodes $s,t\in V$, set $M\subseteq A$ of \emph{vulnerable} arcs, and integer
		$\ell\in \mathbb{Z}_{\geq 0}$.\tabularnewline
		\textbf{Task:}    & Find minimum cost set $S \subseteq A$, such that $S \setminus f$ contains $\ell$ disjoint $s$-$t$ paths for every $f \in M$.
	\end{tabular}
\end{framed}

\subsection{Results}

A well-known polynomial algorithm for EDP works as follows. Assign unit capacities to all edges in $G$ and find a minimum-cost
$k$-flow from $s$ to $t$. The integrality property of the \textit{Minimum-Cost $s$-$t$ Flow} (MCF) problem guarantees that an
extreme-point optimal solution is integral, hence it corresponds to a subset of edges. It is then straightforward to verify that
this set is an optimal solution of the EDP problem (for a thorough treatment of this method we refer to the book of Schrijver~\cite{schrijver2003combinatorial}).

The latter algorithm raises two immediate questions concerning FTP. The first
question is whether FTP admits a polynomial time algorithm. In this paper we
give a negative answer conditioned on $\P\neq\NP$, showing that FTP is
\NP-hard.  In fact, the existence of constant-factor approximation algorithms
is unlikely already for the restricted case of directed acyclic graphs.
Consequently, it is natural to ask whether polynomial algorithms can be
obtained for restricted variants of FTP 
to this question. In particular we provide polynomial-time algorithms for
arbitrary graphs and $k=1$, directed acyclic graphs and fixed $k$ and series-parallel graphs.
A second question
concerns the natural fractional relaxation FRAC-FTP of FTP, 
in which the task is to find a minimum cost capacity vector $x \in [0, 1]^A$
such that for every set $F$ of at most $k$ vulnerable edges, the maximum
$s$-$t$ flow in $G_F = (V, A \setminus F)$, capacitated by $x$, is at least
one.  As we previously observed, one natural relaxation of EDP is the MCF
problem.  This relaxation admits an integrality gap of one, namely the optimal
integral solution value is always equal to the corresponding optimal fractional
value. 
We later show that, in contrast to MCF, the integrality gap of FRAC-FTP is
bounded by $k+1$. Furthermore, we show that this bound is tight, namely that
there exists an infinite family of instances with integrality gap arbitrarily
close to $k+1$.  This result also leads to a simple $(k+1)$-approximation
algorithm for FTP, which we later combine with an algorithm for the case $k=1$
to obtain a $k$-approximation algorithm for \ftp.

The second variant of the EDP we study is \ftf, which asks for
$\ell \geq 1$ disjoint $s$-$t$ paths in the presence of non-uniform faults.
Note that if we consider uniform faults (every edge is vulnerable), an optimal
solution is a minimum-cost $s$-$t$ flow of value $k + \ell$, which can be
computed in polynomial time. We show that, again, the presence of non-uniform
faults makes the problem much harder. In fact, it is as hard to approximate as
\ftp, despite the restriction to single-arc faults
(the same result holds for \ftf on undirected graphs).  On the positive side, we give a
simple polynomial-time $(\ell + 1)$-approximation algorithm for \ftf which
computes a MCF with appropriately chosen capacities. 

Note that our positive results for \ftp imply a polynomial-time algorithm for
\ftf and $\ell = 1$. Hence it is natural to investigate the dependence of the
complexity of \ftf on the number $\ell$ of disjoint paths.  To this end, we fix
$\ell$ and study the corresponding slice \emph{\dorkf} of \ftf. 
Our main result in this setting is a 2-approximation algorithm for \dorkf. In a
nutshell, the algorithm first computes 
minimum-cost $\ell$-flow and then makes the resulting $\ell$ disjoint paths
fault tolerant by solving the corresponding \emph{augmentation problem}.
We solve the augmentation problem by reducing it to a shortest path problem; it
is basically a dynamic programming algorithm in disguise. The reduction is
quite involved: in order to construct the instance of Shortest $s$-$t$-Path, we
solve at most $n^{2\ell}$ instances of the Min-cost Strongly Connected
Subgraphs problem on $\ell$ terminal pairs, all of which can be done in
polynomial time since $\ell$ is fixed.

Given our approximation results, one may wonder whether \dorkf might admit a
polynomial-time algorithm (assuming $\P \neq \NP$, say). An indication in this
direction is that for number of problems with a similar flavor, including
robust paths~\cite{bulk}, robust matchings~\cite{robmatch} or robust spanning
trees~\cite{flex}, hardness results were obtained by showing that the
corresponding augmentation problems are hard. In the light of our results above
this approach does not work for \dorkf. 
On the other hand, we show that such a polynomial-time algorithm for \dorkf
implies polynomial-time algorithms for \onetwotwodst and a special case of
\twokdst. Whether these two problems are \NP-hard or not are long-standing open
questions.

\subsection{Related Work}\label{sec:related}

The shortest path problem is one of \emph{the} classical problems in
combinatorial optimization, and as such, it has received considerable
attenation also in the context of fault tolerance/robustness, see for example
\cite{aissi2005approximationA,buesing_12,DemandR,DemandR2,RRSP,RSP,CompRSP}.
Considering \ftp and \ftf, one of the most relevant notions of robustness is
\emph{bulk-robust}, introduced by Adjiashvili, Stiller and
Zenklusen~\cite{bulk}. Here, we are given a set of \emph{failure scenarios}, that is,
a set of subsets of resources that may fail simultaneously. The task is to find
a minimum-cost subset of the resources, such that a desired property (e.g.,
connectivity of a graph) is maintained, no matter which failure scenario
emerges.  Both \ftp and \ftf are special cases of this model.
Adjiashvili, Stiller and Zenklusen considered bulk-robust counterparts of the
Shortest Path and Minimum Matroid basis problems. For bulk-robust shortest
paths on undirected graphs they give a $O(k + \log n)$-approximation algorithm,
where $k$ is the  maximum size of a failure scenario. However, not that the running-time
of this algorithm is exponential in $k$. Note
that their bulk-robust shortest path problem generalizes \ftp, and therefore
the same approximation guarantee holds for \ftp.
Our approximation algorithm for \ftp significantly improves on this bound, on both
the approximation guarantee and the running-time.

The robustness model used in this paper is natural for various classical
combinatorial optimization problems. Of particular interest is the counterpart
of the Minimum Spanning Tree problem. This problem is closely related to the
Minimum $k$-Edge Connected Spanning Subgraph ($k$-ECSS) problem, a
well-understood robust connection problem. There are numerous results for the
unweighted version of $k$-ECSS.  Gabow, Goemans, Tardos and
Williamson~\cite{kEdgeConnected1} developed a polynomial time $(1 +
\frac{c}{k})$-approximation algorithm for $k$-ECSS, for some fixed constant
$c$. The authors also show that for some constant $c' < c$, the existence of a
polynomial time $(1 + \frac{c'}{k})$-approximation algorithm implies $\P=\NP$.
An intriguing property of $k$-ECSS is that the problem becomes easier to
approximate when $k$ grows. Concretely, while for every fixed $k$, $k$-ECSS is
\NP-hard to approximate within some factor $\alpha_k>1$, the latter result
asserts that there is function $\beta(k)$ tending to one as $k$ tends to
infinity such that $k$-ECSS is approximable within a factor $\beta(k)$.  This
phenomenon was already discovered by Cheriyan and Thurimella~\cite{CT:00}, who
gave algorithms with a weaker approximation guarantee.  The more general
Generalized Steiner Network problem admits a polynomial $2$-approximation
algorithm due to Jain~\cite{Jain2001}. This is also the best known bound for
weighted $2$-ECSS.

Adjiashvili, Hommelsheim and M\"uhlenthaler\cite{flex} considered the bulk-robust minimum spanning tree problem with non-uniform single-edge failures.
Their main result is a $2.523$-approximation algorithm for this problem. 
A problem of a similar flavor but with a uniform fault model is Weighted Robust Matching Augmentation, which
was studied by Hommelsheim, M\"uhlenthaler and Schaudt~\cite{robmatch}. The
task is to find a minimum-cost subgraph, such that after the removal of any
single edge, the resulting graph contains a perfect matching. They show that
this problem is as hard to approximate as \dsnp, which is known to admit no
$\log^{2-\varepsilon}$-approximation algorithm unless $\NP \subseteq
\mathsf{ZTIME}(n^{\polylog(n)})$~\cite{DBLP:conf/stoc/HalperinK03}. We will
later show that \ftf generalizes Weighted Robust Matching Augmentation.

\subsection{Notation}
We mostly consider directed graphs, which we denote by $(V, A)$, where $V$ is
the set of vertices set and $A$ the set of arcs.  Undirected graphs  are
denoted by $(V, E)$, where $E$ is a set of edges. 
An  \emph{orientation} of a set $E$ of undirected edges is an arc-set that
orients each edge $vw \in E$ as an arc $vw$ or $wv$.
For some vertex set $V'
\subseteq V$ we denote by $\delta(V') \coloneqq \{ v w \in E \mid v \in V', w
\in V \setminus V' \}$.  For two vertex sets $X, Y \subseteq V$ we write $E(X,
Y) \coloneqq \{ x y  \in E \mid x \in X \setminus Y, y \in Y \setminus X \}$
for the set of edges joining $X$ and $Y$ (the graph should be clear from the
context). In a directed graph we simply replace $E$ by $A$.
In this paper we usually consider edge-weighted graphs and assume throughout that weights are non-negative.
The arcs of $A$ that are not vulnerable are called \emph{safe}, denoted by 
$\overline{M} \coloneqq A \setminus M$.

For the sake of a clearer presentation, we moved proofs of results marked by $(\ast)$ to the appendix.
A preliminary version of this paper can be found here~\cite{adjiashvili2013fault}.

\subsection{Organization}
The remainder of this paper is organized as follows: We present our results on
the problem \ftp in Section~\ref{sec:ftp} and our results on the problem \ftf
in Section~\ref{sec:ftf}.
In Section~\ref{sec:undirected}, we show that \ftf on undirected graphs is a
special case of \ftf on directed graphs.  We study the approximation hardness
of \ftf in Section~\ref{sec:complexity} and we provide some exact polynomial
algorithms for special cases in Section~\ref{sec:exact}. In
Section~\ref{sec:intgap} we relate \ftp and FRAC-FTP by proving a tight bound
on the interagrality gap and show how this result leads to a $k$-approximation
algorithm for FTP. In Section~\ref{sec:FTF:apx-hardnes} we prove approximation
hardness of \ftf. We then give an $(\ell +1)$-approximation algorithm in
Section~\ref{sec:FTF:apx}, followed by a 2-approximation algorithm for
\ftf with a fixed flow value $\ell$.
Furthermore, in Section~\ref{sec:FTF:other-problems}, we relate the complexity
of \ftf with fixed flow value to other problems of open complexity status.
Section~\ref{sec:conclusions} concludes the paper and mentions some interesting
open problems.

\section{Fault-Tolerant Paths}\label{sec:ftp}

\subsection{Directed Versus Undirected Graphs}\label{sec:undirected}
The classical shortest path problem is set on directed graphs.
Assuming non-negative edge-weights, undirected graphs are a special case, since
we may replace each undirected edge by two antiparallel directed edges and
conclude that any shortest path in the resulting digraph corresponds to a
shortest path in the original undirected graph. Here, we show that the same is
true for \ftp. The main insight is that, even if at most
$k$ vulnerable edges may fail, no undirected edge is used in both directions.
As a consequence, all our positive results for directed graphs in this section also hold
for \ftf  on undirected graphs.

\begin{proposition}
	\label{prop:ftp:directed-undirected}
	Let $X \subseteq E$ be a feasible solution to an instance of \ftp on an
	undirected graph $(V, E)$. Then there is an orientation
	$\overrightarrow{X}$ of $X$ such that $(V, \overrightarrow{X} - F)$
	contains a directed $s$-$t$ path for every $F \subseteq M$ with $|F|
	\leq k$.
\end{proposition}
\begin{proof}
	Let us assume for a contradiction that there is no such orientation.  A
	set $Y$ of (undirected and directed) edges is a \emph{partial orientation
	of $X$} if there is a partition of $X$ into sets $X_1$ and
	$X_2$, such that $Y= X_1 \cup \overrightarrow{X_2}$, where
	$\overrightarrow{X_2}$ is an orientation of $X_2$.  Let $Y$ be a
	partial orientation of $X$ that maximizes the number of
	directed edges, such that $(V, \overrightarrow{X} - F)$ contains a
	directed $s$-$t$ path for every $F \subseteq M$ with $|F| \leq k$. By
	our assumption, there is at least one undirected edge $e = vw$ in $Y$.
	Furthermore, there are two sets $S_1, S_2 \subset V$ of vertices, such
	that $\{s\} \subseteq S_1, S_2 \subseteq V \setminus \{t\}$, $v \in S_1
	\setminus S_2$, and $w \in S_2 \setminus S_1$. Note that $vw \in
	\delta(S_1)$ and $wv \in \delta(S_2)$.

	Since $e$ is needed in both directions for $Y$ to be feasible, there is
	some $F \subseteq M$, $|F| \leq k$, such that $X \setminus F$ contains an
	$s$-$t$ path that must leave $S_1$ via $vw$. Therefore, the cut $\delta(S_1)$ contains at
	most $k+1$ edges and all of them except possibly $e$ are vulnerable. The same holds for $\delta(S_2)$
	and therefore we have $|\delta(S_1)| = |\delta(S_2)| = k+1$. From the
	feasibility of $Y$ and the fact that all edges in $\delta(S_1)$ and
	$\delta(S_2)$ except possibly $e$ are vulnerable, it follows that 
	$|\delta(S_1 \cap S_2)| \geq k+1$ and $|\delta(S_1 \cup S_2)| \geq k+1$.  
	By the submodularity of the cut function $|\delta(\cdot)|$ we have
	\[
		2k+2 = |\delta(S_1)| + |\delta(S_2)| \geq |\delta(S_1 \cap S_2)| + |\delta(S_1 \cup S_2)| \geq 2k+2 
	\]
	and it follows that
	\begin{equation}
		|\delta(S_1)|= |\delta(S_2)|=|\delta(S_1 \cap S_2)| = |\delta(S_1 \cup S_2)| = k+1 \enspace.
		\label{eq:samesize}
	\end{equation}
	The cut-function $|\delta(\cdot)|$ satisfies the following identity
	\[
		|\delta(S_1)| + |\delta(S_2)| =  |\delta(S_1 \cap S_2)| + |\delta(S_1 \cup S_2)| + | A(S_1 \setminus S_2, S_2 \setminus S_1)| + | A(S_2 \setminus S_1, S_1 \setminus S_2)|\enspace, 
	\]
	but the observation that $e$ is an edge connecting $S_1 \setminus S_2$
	and $S_2 \setminus S_1$, together with~\eqref{eq:samesize} yields a
	contradiction to the previous identity.
\end{proof}

\subsection{Complexity of FTP}\label{sec:complexity}

Our first observation is that \ftp generalizes the \textit{Directed $m$-Steiner Tree Problem} ($m$-DST). The input
to $m$-DST is a weighted directed graph $D = (V,A)$, a \textit{source} node $s\in V$, a collections of \textit{terminals}
$T\subseteq V$ and an integer $m \leq |T|$. The goal is to find a minimum-cost arboresence $X \subseteq A$ rooted at $s$, that contains a 
directed path from $s$ to some subset of $m$ terminal. 

The $m$-DST is seen to be a special case of \ftp as follows. Given an instance $I = (D, s,T,m)$ of $m$-DST 
define the following instance of \ftp. The graph $D$ is augmented by $|T|$ new zero-cost arcs $A'$ 
connecting every terminal $u\in T$ to a new node $t$. Finally, we set $M = A'$ and $k = m-1$. 
The goal is to find a fault-tolerant path from $s$ to $t$ in the new graph. It is now straightforward to see
that a solution $S$ to the \ftp instance is feasible if and only if $S\cap A$ contains a feasible solution
to the $k$-DST problem (we can assume that all arcs in $A'$ are in any solution to the \ftp instance).

The latter observation implies an immediate conditional lower bound on the approximability of \ftp. Halperin
and Krauthgamer~\cite{DBLP:conf/stoc/HalperinK03} showed that $m$-DST cannot be approximated within a factor
$\log ^{2-\epsilon} m$ for every $\epsilon > 0$, unless $\NP \subseteq \mathsf{ZTIME}(n^{\polylog(n)})$. As a result
we obtain the following.

\begin{proposition}\label{prop:complexity}
	\ftp admits no polynomial-time approximation algorithms with ratio
	$\log ^{2-\epsilon} k$ for every $\epsilon > 0$, unless $\NP \subseteq
	\mathsf{ZTIME}(n^{polylog(n)})$.
\end{proposition}

	The reduction above can be easily adapted to obtain a
	$k^\epsilon$-approximation algorithm for \ftp for the special case that
	$M \subseteq \{e\in A: t\in e\}$ using the algorithm of Charikar et.
	al.~\cite{charikar1999approximation}.
In fact, any approximation algorithm with factor $\rho(k)$ for \ftp is an
approximation algorithm with factor $\rho(m)$ for $m$-DST. The best known
algorithm for the latter problem is due to Charikar et.
al.~\cite{charikar1999approximation}.  Their result is an approximation scheme
attaining the approximation factor of $m^\epsilon$ for every $\epsilon > 0$.

We end this discussion by showing that \ftp contains a more general Steiner
problem, which we call \textit{Simultaneous Directed $m$-Steiner Tree}
($m$-SDST),  as a special case.  An input to $m$-SDST specifies two arc-weighted
graphs $D_1 = (V,A_1,w_1)$ and $D_2=(V,E_2,w_2)$ on the same set of vertices
$V$, a source $s$, a set of terminals $T\subseteq V$ and an integer $m \leq
|T|$. The goal is to find a subset $U\subseteq T$ of $m$ terminals and two
arboresences $S_1\subseteq E_1$ and $S_2\subseteq A_2$ connecting $s$ to $U$ in
the respective graphs, so as to minimize $w_1(S_1) + w_2(S_2)$. $m$-SDST is
seen to be a special case of \ftp via the following reduction. Given an
instance of $m$-SDST, construct a graph $D = (V',A)$ as follows.  Take a
disjoint union of $D_1$ and $D_2$, where the direction of every arc in $D_2$ is
reversed. Connect every copy of a terminal $u\in T$ in $D_1$ to its
corresponding copy in $D_2$ with an additional zero-cost arc $e_u$. Finally,
set $M = \{e_u : u\in T\}$ and $k = m-1$. A fault-tolerant path connecting the
copy of $s$ in $D_1$ to the copy of $s$ in $D_2$ corresponds to a feasible
solution to the $m$-SDST instance with the same cost, and vice-versa.

\subsection{Polynomial Special Cases}\label{sec:exact}

This section is concerned with tractable restrictions of \ftp. Concretely we give polynomial
algorithms for arbitrary graphs and $k=1$ and directed acyclic graphs (DAGs) and fixed $k$
 and for Series-parallel graphs. 
We denote the problem \ftp restricted to instances with some 
fixed $k$ by $k$-\ftp.

\paragraph*{$1$-\ftp}\label{subsec:kisone}

We start by giving the following structural insight.

\begin{lemma}[$\ast$]\label{lem:disjointpaths}
 Let $X^*$ be an optimal solution to \ftp on the instance $(D,M,k)$.
The minimum
$s$-$t$ flow in the graph $(V, X^*)$ capacitated by the vector $c_e = 1$ if $e\in M$ and $c_e = \infty$,
otherwise is at least $k+1$.
\end{lemma}

An $s$-$t$ \textit{bipath} in the graph $D =(V,A)$ is a union of two $s$-$t$ paths $P_1, P_2 \subseteq A$
In the context of $1$-\ftp, we call a bipath $Q=P_1 \cup P_2$ \textit{robust}, if it holds that $P_1 \cap P_2 \cap M = \emptyset$.
Note that every robust $s$-$t$ bipath $Q$ in $G$ is a feasible solution to the $1$-\ftp instance.
Indeed, consider any vulnerable edge $e\in M$. Since $e\not\in P_1\cap P_2$ it holds that either $P_1 \subseteq Q-e$,
or $P_2 \subseteq Q-e$. It follows that $Q-e$ contains some $s$-$t$ path.
The next lemma shows that every feasible solution of the $1$-\ftp instance contains a robust $s$-$t$ bipath.

\begin{lemma}[$\ast$]\label{lem:robustbipath_contained}
Every feasible solution $S^*$ to an $1$-\ftp instance contains a robust $s$-$t$ bipath.
\end{lemma}

We can conclude from the previous discussion and Lemma~\ref{lem:robustbipath_contained} that all
minimal feasible solutions to the $1$-\ftp instance are robust bipaths. This observations leads
to the simple algorithm, which is given in the proof of the following theorem.

\begin{theorem}[$\ast$]\label{thm:kisone}
	1-\ftp admits a polynomial-time algorithm.
\end{theorem}

\paragraph{$k$-\ftp and Directed Acyclic Graphs}\label{subsec:dagfixedk}

Let us first consider the case of a layered graph. The generalization to a
directed acyclic graph is done via a standard transformation, which we describe later.
Recall that a layered graph 
$D = (V, A)$ is a graph with a partitioned vertex set $V = V_1 \cup \cdots\cup V_r$
and a set of edges satisfying $A \subset \bigcup_{i\in [r-1]} V_i \times V_{i+1}$.
We assume without loss of generality that $V_1 = \{s\}$ and $V_{r} = \{t\}$. For every $i\in [r-1]$ 
we let $A_i = A \cap V_i \times V_{i+1}$.

Analogously to the algorithm in the previous section, we reduce $k$-\ftp to a shortest
path problem in a larger graph. The following definition sets the stage for the
algorithm.

\begin{definition}\label{def:conf}
 An \emph{$i$-configuration} is a vector $d \in \{0,1,\cdots,k+1\}^{V_i}$ 
satisfying $\sum_{v\in V_i} d_v = k+1$. We let $supp(d) = \{v\in V_i : d_v > 0\}$.
For an $i$-configuration $d^1$ and an $(i+1)$-configuration $d^2$ we let 
$$
V(d^1, d^2) = supp(d^1) \cup supp(d^2) \,\, \text{ and } \,\, A(d^1, d^2) = A[V(d^1, d^2)].
$$

We say that an $i$-configuration $d^1$ \emph{precedes} an $(i+1)$-configuration $d^2$ if the 
following flow problem is feasible. The graph is defined as $H(d^1, d^2) = (V(d^1, d^2), A(d^1, d^2))$.
The demand vector $\nu$ and the capacity vector $c$ are given by 
\begin{equation*}
\nu_u = 
\begin{cases} 
-d^1_u    & \mbox{if }  u\in supp(d^1)\\ 
d^2_u      & \mbox{if }  u\in supp(d^2) 
\end{cases} \,\, \text{and} \,\,\,\,
c_e = 
\begin{cases} 
1    	    & \mbox{if }  e\in M\\ 
\infty      & \mbox{if } e\in E\setminus M,
\end{cases}
\end{equation*}
respectively. If $d^1$ precedes $d^2$ we say that the \emph{link} $(d^1,d^2)$ exists. Finally, the \emph{cost}
$\ell(d^1,d^2)$ of this link is set to be minimum value $w(A')$ over all $A'\subseteq A(d^1, d^2)$,
for which the previous flow problem is feasible, when restricted to the set of edges $A'$.
\end{definition}

The algorithm constructs a layered graph $\mathcal{H} = (\mathcal{V}, \mathcal{A})$ with $r$ layers 
$\mathcal{V}_1, \cdots, \mathcal{V}_r$. For every $i\in[r]$ the set of 
vertices $\mathcal{V}_i$ contains all $i$-configurations. Observe that $\mathcal{V}_1$ and $\mathcal{V}_r$ 
contain one vertex each, denoted by $c^s$ and $c^t$, respectively. The edges correspond to links between
configurations. Every edge is directed from the configuration with the lower index to the one with the higher index.
The cost is set according to Definition~\ref{def:conf}.
The following lemma provides the required observation, which immediately leads to a polynomial algorithm.

\begin{lemma}[$\ast$]
	\label{lem:label}
	Every $c^s$-$c^t$ path $P$ in $H$ corresponds to a fault-tolerant path $S$ with $w(S) \leq \ell(P)$, and vise-versa.
\end{lemma}

Finally, we observe that the number of configurations is bounded by $O(n^{k+1})$, which implies that
$k$-\ftp can be solved in polynomial time on layered graphs. 

To obtain the same result for directed
acyclic graphs we perform the following transformation of the graph. Let $v_1, \cdots, v_n$
be a topological sorting of the vertices in $D$. Replace every edge $e = v_iv_j$ ($i<j$) with a path 
$p_e = v_i, u^e_{i+1}, \cdots, u^e_{j-1}, v_j$ of length $j-i+1$ by subdividing it sufficiently many times. Set
the cost of the first edge on the path to $w'(v_iu^e_{i+1}) = w(v_iv_j)$ and set the costs of all other edges on the path
to zero. In addition, create a new set of faulty edges $M'$, which contains all edges in a path $p_e$ if $e\in M$.
It is straightforward to see that the new instance of \ftp is equivalent to the original one, while the obtained
graph after the transformation is layered. We summarize the result as follows.

\begin{theorem}\label{thm:dagpoly}
There is a polynomial algorithm for $k$-\ftp restricted to instances with a directed acyclic graph.
\end{theorem}

\paragraph{Series-Parallel Graphs}\label{subsec:srp}

Recall that a graph is called \emph{series-parallel (SRP)} with terminal $s$ and $t$ if 
it can be composed from a collection of disjoint edges using the \emph{series} and \emph{parallel} compositions.
The series composition of two SRP graphs with terminals $s$, $t$ and $s',t'$ respectively, takes
the disjoint union of the two graphs, and identifies $t$ with $s'$. The parallel composition 
takes the disjoint union of the two graphs and identifies $s$ with $s'$ and $t$ with $t'$.
Given a SRP graph it is easy to obtain the aforementioned decomposition.

The algorithm we present has linear running time whenever the robustness parameter $k$ is fixed.
The algorithm is given as Algorithm~\ref{alg:ercc:sircc:srp}. 
In fact, the algorithms computes the optimal solutions $S_{k'}$ for all parameters
$0 \leq k' \leq k$. The symbol $\perp$ is returned if the problem is infeasible.

\begin{algorithm}
\caption{\textbf{: \ftp-SeriesParallel($G,s,t,k$)}}\label{alg:ercc:sircc:srp}
\begin{algorithmic}[1]
\REQUIRE{$G=(V,E)$ a series-parallel graph, $s,t\in V$ and $M\subset E$, $k\in \mathbb{Z}_{\geq 0}$.}
\ENSURE{Optimal solution to \ftp for parameters $0,1,\cdots, k$.}
\IF{$E = \{e\} \wedge e\in M$}
  \STATE Return $(\{e\}, \perp, \cdots, \perp)$
\ENDIF
\IF{$E = \{e\} \wedge e\not\in M$}
  \STATE Return $(\{e\}, \cdots, \{e\})$
\ENDIF

\vspace{2mm}

$\Rightarrow$ $G$ is a composition of $H_1, H_2$.
\vspace{2mm}

\STATE $(S_0^1, \cdots, S_k^1) \leftarrow $ \textbf{\ftp-SeriesParallel}$(H_1, M\cap E[H_1], k)$
\STATE $(S_0^2, \cdots, S_k^2) \leftarrow $ \textbf{\ftp-SeriesParallel}$(H_2, M\cap E[H_2], k)$

\IF{$G$ \text{is a series composition of} $H_1, H_2$}

\FOR{$i=0,\cdots,k$}
  \IF{$S_i^1 = \perp \vee \, S_i^2 = \perp$}
    \STATE $S_i \leftarrow \perp$
  \ELSE 
    \STATE $S_i \leftarrow S_i^1 \cup S_i^2$
  \ENDIF
\ENDFOR

\ENDIF

\IF{$G$ \text{is a parallel composition of} $H_1, H_2$}
  \STATE $m_1 \leftarrow \max\{i : S_i^1 \neq \perp\}$
  \STATE $m_2 \leftarrow \max\{i : S_i^2 \neq \perp\}$

  \FOR{$i=0,\cdots,k$}
    \IF{$i > m_1 + m_2 + 1$} 
      \STATE $S_i \leftarrow \perp$
    \ELSE 
      \STATE $r \leftarrow argmin_{-1\leq j\leq i}\{w(S_j^1) + w(S_{i-j-1}^2)\}$ \hspace{20mm} $// \, S^1_{-1} = S^2_{-1} := \emptyset$
      \STATE $S_i \leftarrow S_r^1 \cup S_{i-r-1}^2$
    \ENDIF
  \ENDFOR
\ENDIF

\STATE Return $(S_0, \cdots, S_k)$
\end{algorithmic}
\end{algorithm}

\begin{theorem}[$\ast$]\label{thm:fastalgSRP}
Algorithm~\ref{alg:ercc:sircc:srp} returns an optimal solution to the \ftp problem on SRP
graphs. The running time of Algorithm~\ref{alg:ercc:sircc:srp} is $O(nk)$.
\end{theorem}

\subsection{Integrality Gap and Approximation Algorithms}\label{sec:intgap}

In this section we study the natural fractional relaxation of \ftp. We prove a tight
bound on the integrality gap of this relaxation. This results also suggests a simple
approximation algorithm for \ftp with ratio $k+1$. We later combine this algorithm 
with the algorithm for $1$-\ftp to obtain a $k$-approximation algorithm.

\paragraph{Fractional \ftp and Integrality Gap}

Let us start by introducing the fractional relaxation of \ftp, which we denote
by FRAC-FTP. The input to FRAC-FTP is identical to the input to \ftp. The goal
in FRAC-FTP is to find a capacity vector $x:A \rightarrow [0,1]$ of minimum
cost $w(x) = \sum_{e\in A} w_e x_e$ such that for every $F\subseteq M$ of size
at most $k$, the maximum $s$-$t$ flow in $D-F$, capacitated by $x$ is at least
one.  Note that by the Max-Flow Min-Cut Theorem, the latter condition is
equivalent to requiring that the minimum $s$-$t$ cut in $D-F$ has capacity of
at least one. We will use this fact in the proof of the main theorem in this
section.

Observe that by requiring $x\in \{0,1\}^E$ we obtain \ftp, hence FRAC-FTP is indeed a fractional relaxation
of \ftp.

In the following theorem by 'integrality gap' we mean the maximum ratio between the optimal solution
value to an \ftp instance, and the optimal value of the corresponding FRAC-FTP instance.

\begin{theorem}[$\ast$]\label{thm:integralitygap}
	The integrality gap of \ftp is bounded by $k+1$. Furthermore, there exists an infinite family of
	instances of \ftp with integrality gap arbitrarily close to $k+1$.
\end{theorem}

The proof of Theorem~\ref{thm:integralitygap} implies a simple $(k+1)$-approximation algorithm for
\ftp. This algorithm simply solves the integer minimum-cost flow problem, defined in proof of the theorem, and returns 
the set of edges corresponding to the support of an optimal integral flow $z^*$ as the solution. 
This result is summarized in the following corollary.

\begin{corollary}\label{cor:kplusone_pprox}
There is a polynomial $(k+1)$-approximation algorithm for \ftp.
\end{corollary}

\paragraph{A $k$-Approximation Algorithm}\label{subsec:kapprox}

In this paragraph we improve the approximation algorithm from the previous paragraph.
The new algorithm can be seen as a generalization of the algorithm for 1-\ftp to arbitrary \ftp instances. The main
observation is the following. The reason why the approximation algorithm implied by Theorem~\ref{thm:integralitygap}
gives an approximation ratio of $k+1$ is that the capacity of edges in $A \setminus M$ is set to $k+1$, hence,
if the flow $z^*$ uses such edges to their full capacity, the cost incurred is $k+1$ times the cost of these
edges. This implies that the best possible lower bound on the cost $w(z^*)$ is $(k+1)OPT_{FRAC}$, where 
$OPT_{FRAC}$ denotes the optimal solution value of the corresponding FRAC-FTP instance. 
To improve the 
algorithm we observe that the edges in $z^*$, which carry a flow of $k+1$ are cut-edges in the obtained solution.

To conveniently analyze our new algorithm let us consider a certain canonical flow defined 
by minimal feasible solutions.

\begin{definition}\label{def:canonicalflow}
Consider an inclusion-wise minimal feasible solution $S\subseteq A$ of an instance $I = (D,s,t,M)$ of \ftp. 
A \emph{flow $f^S$ induced by $S$} is any integral $s$-$t$ $(k+1)$-flow in $D$ respecting the capacity vector
$$
c^S_e = 
\begin{cases} 
1    & \mbox{if }  e \in S\cap M\\ 
k+1  & \mbox{if }  e \in S \setminus M\\ 
0    & \mbox{if }   e\in A\setminus S.
\end{cases}
$$
\end{definition}

To this end consider an optimal solution $X^* \subseteq A$ to the \ftp instance and consider any corresponding 
induced flow $f^*$. Define 
$$X_{PAR} = \{e\in X^* : f^*(e) \leq k\} \,\,\, \text{and} \,\,\,  X_{BRIDGE} = \{e\in X^* : f^*(e) = k+1\}.$$
As we argued before, every edge in $X_{BRIDGE}$ must be a bridge in $H = (V,X^*)$ disconnecting $s$ and $t$.
Let $u_e$ denote the tail vertex of an edge $e\in A$. Since every edge $e\in X_{BRIDGE}$ constitutes
an $s$-$t$ cut in $H$, it follows that the vertices in $U = \{e_u : e\in X_{BRIDGE}\} \cup \{s,t\}$ can be unambiguously 
ordered according to the order in which they appear on any $s$-$t$ path in $H$, traversed from $s$ to $t$. 
Let $s = u_1, \cdots, u_q = t$ be this order. Except for $s$ and $t$, every vertex in $U$ constitutes
a cut-vertex in $H$. Divide $H$ into $q-1$ subgraphs $H^1, \cdots, H^{q-1}$ by letting $H^i = (V, Y_i)$ contain the 
union of all $u_i$-$u_{i+1}$ paths in $H$. We observe the following property.

\begin{proposition}\label{prop:decomp}
For every $i\in [q-1]$ the set $Y_i\subseteq A$ is an optimal solution to the \ftp instance $I_i = (G,u_i, u_{i+1}, M)$.
\end{proposition}

Consider some $i\in[q-1]$ and let $f^*_i$ denote the flow $f^*$, restricted to edges in $H^i$. Note that
$f^*_i$ can be viewed as a $u_i$-$u_{i+1}$ $(k+1)$-flow.
Exactly one of the following cases can occur. Either $H^i$ contains a single edge $e\in A \setminus M$, or 
$$
\max_{e\in Y_i}{f^*_i(e)} \leq k.
$$
In the former case, the edge $e$ is the shortest $u_i$-$u_{i+1}$ path in $(V,A \setminus M)$. In the latter
case we can use a slightly updated variant of the algorithm in Corollary~\ref{cor:kplusone_pprox} to obtain
a $k$-approximation of the optimal \ftp solution on instance $I_i$. Concretely, the algorithm defines 
the capacity vector
$$
c'_e = 
\begin{cases} 
k    & \mbox{if } e\not\in M \\ 
1  & \mbox{otherwise},
\end{cases}
$$
and finds an integral minimum-cost $u_i$-$u_{i+1}$ $(k+1)$-flow $Y^*$ in $D$, and returns the support $Y\subseteq A$
of the flow as the solution. The existence of the flow $f^*_i$ guarantees that $w(y^*) \leq w(f^*_i)$,
while the fact that the maximum capacity in the flow problem is bounded by $k$ gives $w(Y) \leq kw(y^*)$.
It follows that this algorithm approximates the optimal solution to the \ftp instance $I_i$ to within a factor $k$.

To describe the final algorithm it remains use the blueprint of the algorithm for 1-\ftp. 
There is only one slight difference.
Instead of finding two edge-disjoint $u$-$v$ paths, the new
algorithm solves the aforementioned flow problem. 
We summarize the main result of this section in the following theorem. The proof is omitted, as it is identical
to that of Theorem~\ref{thm:kisone}, with the exception of the preceding discussion.

\begin{theorem}
\label{thm:k-approx}
There is a polynomial $k$-approximation algorithm for \ftp.
\end{theorem}

\section{Fault-Tolerant Flows}\label{sec:ftf}

In this section we present our results on the problem \ftf. We show that it
admits no $\log^{2-\varepsilon} \ell$-approximation under standard complexity
assumptions. We then investigate its complexity for flows of
fixed value $\ell$. Our main result is a polynomial-time algorithm for the
corresponding augmentation problem, which we use to obtain a 2-approximation
for \dorkf. Finally, we show that a polynomial-time algorithm for \dorkf
implies polynomial-time algorithms for two problems whose complexity status is
open.

\subsection{Approximation Hardness of \ftf} \label{sec:FTF:apx-hardnes}

We show that \ftf is as hard to approximate as \dsnp by using an approximation
hardness result from~\cite{robmatch} for the problem \wrma.  The problem \wrma
asks for the cheapest edge-set (assuming non-negative costs) to
add to a bipartite graph such that the resulting graph is bipartite and
contains a perfect matching after a-posteriori removing any single edge. The
idea of our reduction is similar to that of the classical reduction from the
Bipartite Maximum Matching problem to the Max $s$-$t$ Flow problem. Note that
since matchings are required to be perfect, we may assume that both parts of
the input graph have the same size. We add to the bipartite input graph $(U, W,
E$) on $n$ vertices of a \wrma instance two
terminal vertices $s$ and $t$, and connect $s$ to each vertex of $U$ as well as
each vertex of $W$ to  $t$ by an arc. Now we add all possible arcs from $U$ to
$W$, marking those as vulnerable that correspond to an edge in $E$. It is
readily observed that a fault-tolerant $n/2$-flow corresponds to a feasible
solution to the given \wrma instance (after removing all arcs incident to $s$ or
$t$). We thus obtain the following hardness result.

\begin{lemma}[$\ast$]
	\label{lemma:dorkf:hardness}
	A polynomial-time $f(\ell)$ approximation algorithm for \ftf implies a
	polynomial-time $f(n/2)$-approximation algorithm for \wrma, where $n$ is
	the number of vertices in the \wrma instance.
\end{lemma}

We combine Proposition~\ref{lemma:dorkf:hardness} with two results
from~\cite{robmatch} and~\cite{DBLP:conf/stoc/HalperinK03} to obtain the
following approximation hardness result for \ftf.

\begin{theorem}[$\ast$]
	\label{thm:dorkf:inapprox}
	\ftf admits no polynomial-time $\log^{2- \varepsilon}(\ell)$-factor approximation algorithm for every $\varepsilon > 0$, unless $\NP \subseteq \mathsf{ZTIME}(n^{\polylog(n)})$.
\end{theorem}

Note that all results presented in this section also hold for the undirected variant of \ftf.

\subsection{Approximation Algorithms}\label{sec:FTF:apx}

We first present a simple polynomial-time $(\ell+1)$-approximation algorithm
for \ftf, which is very similar to the $(k+1)$-approximation for \ftp.  The
algorithm computes (in polynomial time) a minimum-cost $s$-$t$ flow of value
$\ell+1$ on the input graph with the following capacities: each vulnerable arc
receives capacity $1$ and any other arc capacity $1 + 1/\ell$. To see that for
this choice of capacities we obtain a feasible solution, recall that the value
of any $s$-$t$ cut upper-bounds the value of any $s$-$t$ flow. Therefore, each
$s$-$t$ cut $C$ has value at least $\ell+1$, so $C$  contains either at least
$\ell$ safe arcs or at least $\ell+1$ arcs. To prove the
approximation guarantee, we show that any optimal solution to an \ftf instance
contains an $s$-$t$ flow of value $\ell+1$ and observe that we over-pay for
safe arcs by a factor of at most $(1 + 1/\ell)$. We obtain the following
result.

\begin{theorem}[$\ast$]
	\ftf admits a polynomial-time $(\ell+1)$-factor approximation algorithm.
	\label{thm:ftf:approx}
\end{theorem}

Note that 
we cannot simply use the dynamic programming approach as in the algorithm for
1-\ftp to obtain an $\ell$-approximation for \ftf, since a solution to \ftf in
general does not have cut vertices, which are essential for the decomposition
approach for the $k$-approximation for \ftp.

We now show that for a fixed number $\ell$ of disjoint paths, a much better
approximation guarantee can be obtained. That is, we give a polynomial-time
2-approximation algorithm for \dorkf (however, its running time is exponential
in $\ell$). The algorithm first computes a minimum-cost $s$-$t$ flow of value
$\ell$ and then augments it to a feasible solution by solving the following
\emph{augmentation problem}.

\begin{framed}
	\begin{tabular}[h]{lp{0.8\linewidth}}
		\multicolumn{2}{l}{\onerkdpa} \tabularnewline
		\textbf{Instance:} &  arc-weighted directed graph
		$D=(V, A)$, two nodes $s,t\in V$, arc-set $X_0 \subseteq A$ that
		contains $\ell$ disjoint $s$-$t$ paths, and set $M\subseteq A$ of
		vulnerable arcs.\tabularnewline
		\textbf{Task:}    & Find minimum weight set $S \subseteq A
		\setminus X_0$, such that for every $f \in M$, the set $(X_0
		\cup S) \setminus f$ contains $\ell$ disjoint $s$-$t$ dipaths.
	\end{tabular}
\end{framed}

Our main technical contribution is that \onerkdpa can be solved in
polynomial time for fixed $\ell$. Our algorithm is based on a dynamic programming
approach and it involves solving many instances of the problem \dsnp, which
asks for a cheapest subgraph connecting $\ell$ given terminal pairs. This
problem admits a polynomial-time algorithm for fixed $\ell$~\cite{feldman_ruhl_06}, but it is 
$\W[1]$-hard when parameterized in the number of terminal pairs, so it is likely
not fixed-parameter tractable~\cite{guo2011parameterized}. Roughly speaking, we traverse the $\ell$
disjoint $s$-$t$ paths computed previously in parallel, proceeding one arc
at a time. In order to deal with vulnerable arcs, at each step, we solve an
instance of \dsnp connecting the $\ell$ current vertices (one on each path) to
$\ell$ destinations on the same path by using backup paths. That is, we
decompose a solution to the augmentation problem into instances of \dsnp
connected by safe arcs. An optimal decomposition yields an optimal solution to
the instance of the augmentation problem. We find an optimal decomposition by
dynamic programming. Essentially, we give a reduction to a shortest path
problem in a graph that has exponential size in $\ell$.

Let us fix an instance $I$ of \onerkdpa on a digraph $D =
(V, A)$ with arc-weights $c \in \Z_{\geq 0}^A$ and and terminals $s$ and $t$.
Let $P_1, P_2, \ldots, P_\ell$ be $\ell$ disjoint $s$-$t$ paths contained in
$X_0$. In fact, we assume without loss of generality, that $X_0$ is the union
of $P_1, P_2, \ldots, P_\ell$. If $X_0$ contains an arc $e$ that is not on any
of the $\ell$ paths, we remove $e$ from $X_0$ and assign to it weight 0. 

We now give the reduction to the shortest path problem.
We construct a digraph $\mathcal{D}= (\mathcal{V}, \mathcal{A})$;
to distinguish it clearly from the graph $D$ of $I$, we call the elements in $V$ ($A$) of $D$ \emph{vertices} (\emph{arcs})
and elements of $\mathcal{V}$ ($\mathcal{A}$) \emph{nodes} (\emph{links}).
We order the vertices of each path $P_i$, $1 \leq i \leq \ell$, according to
their distance to $s$ on $P_i$. For two vertices $x_i^1, x_i^2$ of $P_i$, we
write $x_i^1 \leq x_i^2$ if $x_i^1$ is at least as close to $s$ on $P_i$ as $x_i^2$.
Let us now construct the node set $\mathcal{V}$.  We add a node $v$ to
$\mathcal{V}$ for every $\ell$-tuple $v = (x_1, ..., x_\ell)$ of vertices in
$V(X_0)$ satisfying $x_i \in P_i$, for every $i \in \{ 1, 2, \ldots, \ell\}$.
Note that the corresponding vertices of a node are not necessarily distinct,
since the $\ell$ \emph{edge-}disjoint paths $P_1, P_2, \ldots, P_\ell$ may
share vertices.
We also define a (partial) ordering on the nodes in $\mathcal{V}$.  For two
nodes $v_1= (x_1^1, ..., x_\ell^1)$ and $v_2= (x_1^2, ..., x_\ell^2)$ we write
$v_1 \leq v_2$ if $x_i^1 \leq x_i^2$ for every $1 \leq i \leq \ell$.
Additionally, let $Q_i(x, y)$ be the sub-path of $P_i$ from a vertex $x$ to a
vertex $y$ of $P_i$.

We now construct the link set $\mathcal{A} := \mathcal{A}_1 \cup \mathcal{A}_2$
of $\mathcal{D}$ as the union of two link-sets $\mathcal{A}_1$ and
$\mathcal{A}_2$, which we will define next.  We add to $\mathcal{A}_1$ an arc
$xy$, if $x$ precedes $y$ and the subpaths of each $P_i$ from $x_i$ to $y_i$
contain no vulnerable arc. That is, we let
\[
	\mathcal{A}_1 := \{ xy \mid x, y \in \mathcal{V},\, x \leq y,\, Q_i(x_i, y_i) \cap M = \emptyset \text{ for $1 \leq i \leq \ell$}\}\enspace.
\]
We now define the link set $\mathcal{A}_2$.  For two nodes $x,y \in
\mathcal{V}$ such that  $x$ precedes $y$, if there is some $1 \leq i\leq \ell$,
such that $Q_i(x_i, y_i)$ contains at least one vulnerable arc, then we first need
to solve an instance of \dsnp on $\ell$ terminal pairs in order to know whether
we add the link $xy$ and, if so, at which cost.  We construct an instance $I(x,
y)$ of \dsnp as follows. The terminal pairs are $(x_i, y_i)_{1 \leq i \leq
\ell}$.  The input graph is given by $D'= (V, A')$, where $A' = (A \setminus X_0) \cup
\bigcup_{1 \leq i  \leq \ell} \overleftarrow{Q_i}(x_i, y_i)$, where
$\overleftarrow{Q_i}(x_i, y_i)$ are the arcs of $Q_i(x_i, y_i)$ in reversed direction
The arc costs are given by
\[
	c'_e \coloneqq
	\begin{cases}
		c_e &   \text{if $e \in A \setminus X_0$, and}\\
		0   &   \text{if $e \in \overleftarrow{Q_i}(x_i, y_i)$ for some  $i \in \{1, 2, \ldots, \ell\}$}. 
	\end{cases}
\]
That is, for $1 \leq i \leq \ell$, we reverse the path $Q_i(x_i, y_i)$ connecting
$x_i$ to $y_i$ and make the corresponding arcs available at zero cost. 
We then need to connect $x_i$ to $y_i$ without using arcs in $X_0$.
Since the number of terminal pairs is at most $\ell$ and thus constant, 
the \dsnp instance $I(v_1, v_2)$ can be solved in polynomial time by the algorithm 
of Feldman and Ruhl given in~\cite{feldman_ruhl_06}. 
Let $\opt(I (x, y))$ be the cost of an optimal solution to $I(v, v)$.
We add a link $xy$ to $\mathcal{A}_2$ if the computed solution of $I(x, y)$
is strongly connected. This completes the construction of $\mathcal{A}_2$. 

For a link $e \in \mathcal{A}$ we let the weight $w_e$ be given by
    \[
        w_e \coloneqq
        \begin{cases}
		0 		&   \text{if $e \in \mathcal{A}_1$, and}\\
		\opt (I(x, y))	&   \text{if $e \in \mathcal{A}_2$.}
        \end{cases}
    \]

We now argue that a shortest path $\mathcal{P}$ from node $s_1 = (s, \ldots, s)
\in \mathcal{V}$ to node $t_1 = (t, \ldots, t) \in \mathcal{V}$ in
$\mathcal{D}$ corresponds to an optimal solution to $I$.  For every link $xy
\in \mathcal{P}$, we add the optimal solution to $I(x, y)$ computed by the
Feldman-Ruhl algorithm to our solution $Y$. A summary is given in
Algorithm~\ref{ALG:DP:augprob}. Proving that Algorithm~\ref{ALG:DP:augprob} is
quite technical and requires another auxiliary graph and a technical lemma. The
details can be found in Appendix~\ref{app:FTF:apx}.

\begin{algorithm}[t]
	\caption{: Exact algorithm for \onerkdpa\label{ALG:DP:augprob}}
	\begin{algorithmic}[1]    
		\REQUIRE{instance $I$ of \onerkdpa on a digraph $D = (V, A)$} 
		\STATE{Construct the graph $\mathcal{D}= ( \mathcal{V}, \mathcal{A})$}
		\STATE{Find a shortest path $\mathcal{P}$ in $\mathcal{D}$ from $(s, \ldots, s)$ to $(t, \ldots, t)$}
		\STATE{For each link $v w \in \mathcal{P} \cap \mathcal{A}_2$ add the arcs of an optimal solution to $I(v, w)$ to $Y$}
		\RETURN $Y$
	\end{algorithmic}
\end{algorithm} 


\begin{theorem}[$\ast$]
	\label{thm:ALG:DP:optimal}
	The set $Y$ computed by Algorithm \ref{ALG:DP:augprob} is an optimal solution  to the instance $I$ of \onerkdpa.
\end{theorem}

Algorithm~\ref{ALG:DP:augprob} runs in polynomial time for a fixed number
$\ell$ of disoint $s$-$t$ paths, since it computes at most $n^\ell$ Min-cost
Strongly Connected Subgraphs on $\ell$ terminal pairs, which can be done in
polynomial time by a result of Feldman and Ruhl~\cite{feldman_ruhl_06}. 

\begin{theorem}[$\ast$]
	\label{thm:ALG:DP:runtime}
	Algorithm~\ref{ALG:DP:augprob} runs in time
	$O(|A| |V|^{6\ell-2} + |V|^{6\ell-1} \log |V|)$.
\end{theorem}

From theorems~\ref{thm:ALG:DP:optimal} and~\ref{thm:ALG:DP:runtime} we obtain a
polynomial-time 2-approximation algorithm for \dorkf: Let $\opt(I)$ be the cost
of an optimal solution to an instance $I$ of \dorkf.  
The algorithm first computes a minimum-cost $s$-$t$ flow $X_0$ and
then runs Algorithm~\ref{ALG:DP:augprob} using $X_0$ as initial arc-set. The
algorithm returns the union of the arc-sets computed in the two steps.  By
Theorem \ref{thm:ALG:DP:optimal} we can augment $X_0$ in polynomial time to a
feasible solution $X_0 \cup Y$ to $I$. Since we pay at most $\opt(I)$ for the
sets $X_0$ and $Y$, respectively, the total cost is at cost at most $2\opt(I)$.
\begin{corollary}
  \label{cor:ALG:DP:2APX}
  \dorkf admits a polynomial-time 2-factor approximation algorithm.
\end{corollary}

\subsection{Relation to Other Problems of Open Complexity}
\label{sec:FTF:other-problems}
In the previous section we showed that there is a polynomial-time algorithm for
\onerkdpa, from which we obtained a 2-approximation for \dorkf. Ideally, one
would like to complement such an approximation result with a hardness or
hardness-of-approximation result. Since \onerkdpa admits a polynomial-time
algorithm according to Theorem~\ref{thm:ALG:DP:runtime}, we cannot use the
augmentation problem in order to prove \NP-hardness of \dorkf; an approach that
has been used successfully for instance for robust paths~\cite{bulk}, robust
matchings~\cite{robmatch} and robust spanning trees~\cite{flex}. Hence, there
is some hope that \dorkf might actually be polynomial-time solvable.
However, we show that a polynomial-time algorithm for \dortf implies
polynomial-time algorithms for two other problems with unknown complexity
status, namely \onetwotwodst and a special case of \twokdst.

We will first consider the relation of \dorkf and \twokdst, which asks for two
disjoint directed paths connecting a root vertex with each terminal: 

\begin{framed}
	\begin{tabular}[h]{lp{0.8\linewidth}}
		\multicolumn{2}{l}{\twokdst} \tabularnewline
		\textbf{Instance:} & directed graph $D=(V, A)$, cost function $c \in \mathbb{Q}^A$, root $s \in V$, and $k$ terminal vertices $t_1, t_2, \dots, t_k \in V$  \tabularnewline
		\textbf{Task:}  & find a minimum-cost set of edges $X \subseteq A$, such that $(V, X)$ contains two edge-disjoint $s$-$t$ paths for each $t \in \{t_1, t_2, \dots, t_k \}$.
	\end{tabular}
\end{framed}

We will denote the set of terminals by $T \coloneqq \{t_1, t_2, \dots, t_k \}$.
According to \cite{cheriyan2014approximating}, even the complexity of
\onetwotwodst is open, which is the following variant of \twokdst: we have only
two terminals $t_1$ and $t_2$ and aim to find two disjoint $s$-$t_1$ paths and
one $s$-$t_2$ path of minimal total cost.
Note that \twokdst is a generalization of \dstp and therefore does not admit a
polynomial-time $\log^{2- \varepsilon} n$ approximation algorithm unless $\NP
\subseteq \mathsf{ZTIME}(n^{\polylog(n)})$~\cite{DBLP:conf/stoc/HalperinK03}.
However, there is a big gap between the complexity of \dstp and \twokdst if the
number $k$ of terminals is fixed.  While it is known that \dstp is
fixed-parameter tractable when parameterized by the number of terminals (and
therefore polynomial-time solvable for constant $k$), it is unknown whether
\twokdst admits a polynomial-time algorithm even for $k=2$ (for $k=1$, an
optimal solution is a minimum-cost 2-flow).  We now show that a special case of
\dorkf corresponds to \twokdst with the additional constraint that every
$s$-$T$ cut contains at least $k+1$ edges.

\begin{proposition}[$\ast$]
	\label{prop:dorkf-twokdst}
	A polynomial-time algorithm for \dorkkf implies a polynomial-time algorithm
	for \twokdst with the additional constraint that every $s$-$T$ cut contains
	at least $k+1$ edges.
\end{proposition}

Furthermore, we show that \onetwotwodst is a special case of \dortf.

\begin{proposition}[$\ast$]
	\label{prop:dortf-onetwotwodst}
	A polynomial-time algorithm for \dortf implies a polynomial-time
	algorithm for \onetwotwodst.
\end{proposition}

\section{Conclusions and Future Work}\label{sec:conclusions}

This paper presents two problems, \ftp and \ftf, which add a non-uniform
fault model to the classical edge-disjoint paths problem.  In this model, not
all $k$-subsets of edges can be removed from the graph after a solution is
chosen, but rather a subset of \emph{vulnerable edges}, which are provided as
part of the input.  Such an adaptation is natural from the point of view of
many application domains.
We observed a dramatic increase in the computational complexity 
due to the fault model with respect to EDP. At the same time we identified several classes of
instances admitting a polynomial exact algorithm.  These classes include the
case $k=1$, directed acyclic graphs and fixed $k$ and series-parallel graphs.
Next, we defined a fractional counterpart of FTP and proved a tight bound on
the corresponding integrality gap. This result lead to a $k$-approximation
algorithm for FTP.  For \ftf, our main results are a $(\ell+1$)-approximation
algorithm and a 2-approximation algorithm for fixed~$\ell$.

One of the main tasks that remains is to improve the understanding of the
approximability of \ftp. In particular, it is interesting to see if the
approximation guarantee for \ftp can be improved to the approximation
guarantees of the best known algorithms for the Steiner Tree problem.  It is
also interesting to relate \ftp to more general problems such a Minimum-Cost
Fixed-Charge Network Flow and special cases thereof.  The complexity of $k$-FTP
in still unknown. It is interesting to see if the methods employed in the
current paper for $1$-FTP and $k$-FTP on directed acyclic graphs can be
extended to $k$-FTP on general graphs.  Another intriguing open question is
whether \dorkf is \NP-hard, which is open even for $\ell=2$. We showed that a
positive result in this direction implies polynomial-time algorithms for two
Steiner problems whose complexity status is open.

\bibliographystyle{plain}
\bibliography{bibliography}

\appendix

\section{Proofs Omitted from Section~\ref{sec:exact}}

\begin{proof}[Proof of Lemma~\ref{lem:disjointpaths}]
Assume the statement is not true. Then, by the Max-Flow Min-Cut Theorem there is some capacitated cut 
$\delta (V')$ for some $V' \subseteq V$ with $s \in V'$ and $t \notin V'$ such that
$c(\delta(V')) < k+1$. By the definition of $c$, this implies that $\delta(V')$ does not contain safe edges.
But then $F \coloneqq \delta(V')$ is a cut in $(V, X^*)$ of size at most $k$, a contradiction.
\end{proof}

\begin{proof}[Proof of Lemma~\ref{lem:robustbipath_contained}]
We assume without loss of generality that $S^*$ is a minimal feasible solution with respect
to inclusion. Let $Y \subseteq S^*$ be the set of bridges in $(V, S^*)$. From feasibility of 
$S^*$, we have $Y \cap M = \emptyset$. Consider any $s$-$t$ path $P$ in $S^*$. Let $u_1, \cdots, u_r$ be
be the set of vertices incident to $Y = P\cap Y$. Let $u_i$ and $u_{i+1}$ be such that $u_iu_{i+1} \not\in Y$.
(if such an edge does not exist, we have $Y = P$, which means that $P$ is a robust $s$-$t$ bipath).
Note that $S^*$ must contain two edge-disjoint $u_i$-$u_{i+1}$ paths $L_1, L_2$. Taking as
the set $Y$ together with all such pairs of paths $L_1, L_2$ results in a robust bipath.
\end{proof}

\begin{proof}[Proof of Theorem~\ref{thm:kisone}]
To solve $1$-\ftp we need to find the minimum cost robust $s$-$t$ bipath. To this end 
let us define two length functions $\ell_1, \ell_2 : V\times V \rightarrow \mathbb{R}_{\geq 0}$. For two
vertices $u,v\in V$ let $\ell_1(u,v)$ denote the shortest path distance from $u$ to $v$ in
the graph $(V,A\setminus M)$, and let $\ell_2(u,v)$ denote the cost of the shortest pair of edge-disjoint
$u$-$v$ paths in $G$. 
Clearly, both length functions can be computed in polynomial time (e.g. using
flow techniques). Finally, set $\ell(u,v) = \min\{\ell_1(u,v), \ell_2(u,v)\}$. 
Construct the complete graph on the vertex set $V$ and associate the length function $\ell$ with it. 
Observe that by definition of $\ell$, any $s$-$t$ path in this graph corresponds to a robust
$s$-$t$ bipath with the same cost, and vice versa. It remains to find the shortest $s$-$t$ bipath
by performing a single shortest $s$-$t$ path in the new graph. For every edge $uv$ in this 
shortest path, the optimal bipath contains the shortest $u$-$v$ path in $(V,a\setminus M)$ if
$\ell(u,v) = \ell_1(u,v)$, and the shortest pair of $u$-$v$ paths in $G$, otherwise.
\end{proof}

\begin{proof}[Proof of Lemma~\ref{lem:label}]
Consider first a fault-tolerant path $S \subseteq A$. We construct a corresponding $c^s$-$c^t$ path in $H$
as follows. Consider any $k+1$ $s$-$t$ flow $f^S$, induced by $S$. Let $p^1, \cdots, p^l$ be a path 
decomposition of $f^S$ and let $1 \leq \rho_1, \cdots, \rho_l \leq k+1$ (with $\sum_{i\in [l]} \rho_i = k+1$) be the
corresponding flow values.

Since $D$ is layered, the path $p^j$ contains exactly one vertex $v^j_i$ from $V_i$ and one edge $e^j_i$ 
from $A_i$ for every $j\in [l]$ and $i\in [r]$. For every $i\in[r]$ define the $i$-configuration $d^i$ with
$$
d^i_v = \sum_{j\in[l] : v = v^i_j} \rho_i,
$$
if some path $p^j$ contains $v$, and $d^i_v = 0$, otherwise. The fact that $d^i$ is an $i$-configuration follows
immediately from the fact that $f^S$ is a $(k+1)$-flow. In addition, for the same reason 
$d^i$ precedes $d_{i+1}$ for every $i\in [r-1]$. From the latter observations and the fact that $d^1 = c^s$
and $d^r =  c^t$ it follows that $P = d^1, d^2, \cdots, d^r$ is a $c^s-c^t$ path in $H$ with cost $\ell(P) \leq w(S)$.

Consider next an $c^s-c^t$ path $P = d^1, \cdots, d^r$ with cost $\ell(P) = \sum_{i=1}^{r-1} \ell(d^i, d^{i+1})$.
The cost $\ell(d^i, d^{i+1})$ is realized by some set of edges $R_i \subseteq E(d^i, d^{i+1})$ for every $i\in [r-1]$.
From Definition~\ref{def:conf}, the maximal $s$-$t$ flow in the graph $D' = (V, R)$ is at least $k+1$, where 
$R = \cup_{i\in [r-1]} R_i$. Next, Lemma~\ref{lem:disjointpaths} guarantees that there exists some feasible solution
$S \subseteq R$, the cost of which is at most $\ell(P)$. In the latter claim we used the disjointness of the sets 
$R_i$, which is due the layered structure of the graph $G$. This concludes the proof of the lemma.
\end{proof}

\begin{proof}[Proof of Theorem~\ref{thm:fastalgSRP}]
 The proof of correctness is by induction on the depth of the recursion in Algorithm~\ref{alg:ercc:sircc:srp}. Clearly 
the result returned by Algorithm~\ref{alg:ercc:sircc:srp} in lines $1$-$6$ is optimal.
Assume next that the algorithm computed correctly all optimal solutions for the subgraphs $H_1, H_2$,
namely that for every $i\in [2]$ and $j\in [k]$, the set $S_j^i$ computed in lines $7$-$8$ is an
optimal solution to the problem on instance $\mathcal{I}_i^j = (H_i, M \cap E[H_i], j)$.

Assume first that $G$ is a series composition of $H_1$ and $H_2$, and let $0 \leq i \leq k$. If
either $S_i^1 = \perp$ or $S_i^2 = \perp$ the problem with parameter $i$ is clearly also infeasible,
hence the algorithm works correctly in this case. Furthermore, since $G$ contains a cut vertex (the
terminal node, which is in common to $H_1$ and $H_2$), a solution $S$ to the problem is feasible for
$G$ if and only if it is a union of two feasible solutions for $H_1$ and $H_2$. From the inductive
hypothesis it follows that $S_i$ is computed correctly in line $14$.

Assume next that $G$ is a series composition of $H_1$ and $H_2$. Consider any feasible solution $S'$
to the problem on $G$ with parameter $i$. Let $S'_1$ and $S'_2$ be the restrictions of $S'$ to edges
of $H_1$ and $H_2$ respectively, and let $n_1$ and $n_2$ be the maximal integers such that
$S'_1$ and $S'_2$ are robust paths for $H_1$ and $H_2$ with parameters $n_1$ and $n_2$, respectively.
Observe that $i \leq n_1 + n_2 + 1$ must hold. Indeed if this would not be the case, then taking
any cut with $n_1 + 1$ edges in $S'_1$ and another cut with $n_2 + 1$ edges in $S'_2$ yields a cut
with $n_1 + n_2 + 2$ edges in $G$, contradicting the fact that $S'$ is a robust path with parameter $i$.
We conclude that the algorithm computes $S_i$ correctly in line $23$. Finally note
that the union any two robust paths for the graphs $H_1$ and $H_2$ with parameters $n_1$ and $n_2$ with 
$i \leq n_1 + n_2 + 1$ yield a feasible solution $S_i$. It follows that the minimum cost such robust
path is obtained as a minimum cost of a union of two solutions for $H_1$ and $H_2$, with robustness
parameters $j$ and $i-j-1$ for some value of $j$. To allow $S_i = S_i^1$ or $S_i = S_i^2$ we let 
$j$ range from $-1$ to $k$ and set $S^1_{-1} = S^2_{-1} = \emptyset$. This completes the proof of correctness.

To prove the bound on the running time, let $T(m,k)$ denote the running time of the algorithm on
a graph with $m$ edges and robustness parameter $k$. We assume that the graph is given by a 
hierarchical description, according to its decomposition into single edges.
The base case obviously takes $O(k)$ time. Furthermore we assume that the solution $(S_0, \cdots, S_k)$
is stored in a data structure for sets, which uses $O(1)$ time for generating empty sets and for
performing union operations.
If the graph is a series composition then the running time satisfies 
$T(m,k) \leq T(m',k) + T(m-m',k) + O(k)$ for some $m' < m$. If the graph is a parallel composition, then $T(m,k)$
satisfied the same inequality. We assume that the data structure, which stores the sets $S_i$
also contains the cost of the edges in the set. This value can be easily updates in time $O(1)$
when the assignment into $S_i$ is performed. It follows that $T(m,k) = O(mk) = O(nk)$ as required.
\end{proof}

\section{Proofs Omitted from Section~\ref{sec:intgap}}

\begin{proof}[Proof of Theorem~\ref{thm:integralitygap}]
	Consider an instance $I = (D,M,k)$ of \ftp. Let $x^*$ denote an optimal solution to the corresponding
	FRAC-FTP instance, and let $OPT = w(x^*)$ be its cost. Define a vector $y\in \mathbb{R}^A$ as follows.
	\begin{equation}\label{eq:gapproof1}
		y_e = 
		\begin{cases} 
			(k+1) x_e    & \mbox{if } e\not\in M \\ 
			\min\{1, (k+1)x_e\}  & \mbox{otherwise}.
		\end{cases}
	\end{equation}
	Clearly, it holds that $w(y) \leq (k+1)OPT$. 
	We claim that every $s$-$t$ cut in $D$ with capacities $y$ has capacity of at least $k+1$.
	Consider any such cut $C \subset A$, represented as the set of edges in the cut. Let 
	$M' = \{e\in M : x^*_e \geq \frac{1}{k+1}\}$ denote the set of faulty edges attaining high fractional
	values in $x^*$. Define $C' = C\cap M'$. If $|C'| \geq k+1$ we are clearly done. Otherwise, assume $|C'| \leq k$.
	In this case consider the failure scenario $F = C'$. Since $x^*$ is a feasible solution it must hold that
	\begin{equation}\label{eq:gapproof2}
		\sum_{e\in C\setminus C'} x^*_e \geq 1.
	\end{equation}
	Since for every edge $e\in C \setminus C'$ it holds that $y_e = (k+1)x^*_e$ we obtain
	\begin{equation}\label{eq:gapproof3}
		\sum_{e\in C\setminus C'} y_e \geq k+1,
	\end{equation} 
	as desired. From our observations it follows that the maximum flow in $D$ with capacities $y$ is at least
	$k+1$. Finally, consider the minimum cost $(k+1)$-flow $z^*$ in $D$ with capacities defined by
	\begin{equation}\label{eq:gapproof4}
		c_e = 
		\begin{cases} 
			k+1    & \mbox{if } e\not\in M \\ 
			1  & \mbox{otherwise}.
		\end{cases}
	\end{equation}
	From integrality of $c$ and the minimum-cost flow problem we can
	assume that $z^*$ is integral. Note that $y_e \leq c_e$ for every $e\in A$, hence any feasible $(k+1)$-flow with capacities
	$y$ is also a feasible $(k+1)$-flow with capacities $c$. From the previous observation it holds that $w(z^*) \leq w(y) \leq (k+1)OPT$.
	From Lemma~\ref{lem:disjointpaths} we know that $z^*$ is a feasible solution to the \ftp instance.
	This concludes the proof of the upper bound of $k+1$ for the integrality gap. 

	To prove the same lower bound we provide an infinite family of instances, containing instances with integrality
	gap arbitrarily close to $k+1$. Consider a graph with $p \gg k$ parallel edges with unit cost connecting $s$ and $t$, and let $M=A$. 
	The optimal solution to \ftp
	on this instance chooses any subset of $k+1$ edges. At the same time, the optimal solution to FRAC-FTP assigns
	a capacity of $\frac{1}{p-k}$ to every edge. This solution is feasible, since in every failure scenario, the number
	of edges that survive is at least $p-k$, hence the maximum $s$-$t$ flow is at least one. The cost of this solution is $\frac{p}{p-k}$.
	Taking $p$ to infinity yields instances with integrality gap arbitrarily close to $k+1$.
\end{proof}

\section{Proofs Omitted from Section~\ref{sec:FTF:apx-hardnes}}

\begin{proof}[Proof of Lemma~\ref{lemma:dorkf:hardness}]
	In the following it will be convenient to denote by $\overline{E}$ the
	edge-set of the bipartite complement of a bipartite graph with edge-set
	$E$.  Let $I = (G, c)$ be an instance of \wrma where $G=(U, W, E)$ is a
	balanced bipartite graph on $n$ vertices and $c \in \Z_{\geq
	0}^{\overline{E}}$. Our reduction is similar to the classical reduction
	from the perfect matching problem in bipartite graphs to the Max
	$s$-$t$ Flow problem.  We construct in polynomial-time an instance $I'
	= (D', c', s, t, M)$ of \ftf as follows.
	To obtain the digraph $D' = (V, A)$, we add to the vertex set of $G$ two new
	vertices $s$ and $t$ and add all arcs from $s$ to $U$ and from $W$ to
	$t$. Furthermore, we add all arcs from $U$ to $W$ and consider those
	that correspondond to an edge in $E$ as vulnerable.
	That is, we let $M \coloneqq \{ u w : u \in U,\, w \in W,\, u w \in E \}$. To
	complete the construction of $I'$, we let $\ell = n/2$, and let the arc-costs $c'$ be given by
	\[
		c'_{u w} \coloneqq
		\begin{cases}
			c_{u w}	&	\text{if $uw \in E(G)$, and}\\
			0 	&	\text{otherwise}.
		\end{cases}
	\] 
	
	For $X \subseteq E \cup \overline{E}$ we write $q(X)$ for the corresponding set of 
	arcs of $D'$. Similarly, for a set $Y \subseteq A$ of arcs we write
	$q^{-1}(Y)$ for the corresponding set of undirected edges of $G$.
	Observe that for a feasible solution $X$ to $I$, the arc set $q(X) \cup A_s \cup
	A_t$ is feasible for $I'$, where $A_s$ (resp., $A_t$) is the set of
	arcs leaving $s$ (resp., entering $t$).  Furthermore, a feasible
	solution $Y$ to $I'$ corresponds to a feasible solution $q^{-1}(Y
	\setminus (A_s \cup A_t))$ to $I$. Also note that, by the choice of
	$c'$, we have that the cost of two corresponding solutions is the same.
	It follows that since $\ell= n/2$, any polynomial-time $f(\ell)$-factor
	approximation algorithm for \dorkf implies a polynomial-time
	$f(n/2)$-factor approximation algorithm for \wrma, where $n = |U + W|$.
\end{proof}

\begin{proof}[Proof of Theorem~\ref{thm:dorkf:inapprox}]
	We give a polynomial-time cost-preserving reduction from \dsnp to \ftf
	via \wrma. The intermediate reduction step from \dsnp to \wrma is given
	in~\cite[Prop.~18]{robmatch}. Consider an instance $I$ of \dsnp on a
	weighted digraph $D = (V, A)$ on $n$ vertices with $k$ terminal pairs
	$(s_1, t_1), (s_2, t_2), \ldots, (s_k, t_k)$. According to the
	reduction given in the proof of~\cite[Prop.~18]{robmatch}, we obtain an
	instance of \wrma on a graph of at most $2(n+k) + 2(n-k) = 4n =: n'$
	vertices. By the arguments their proof, a $f(n')$-approximation
	algorithm for \wrma yields a $f(4n)$-approximation algorithm for \dsnp.
	We apply Proposition~\ref{lemma:dorkf:hardness} to conclude that an
	$f(\ell)$-approximation algorithm for \ftf yields a
	$f(2n)$-approximation algorithm for \dsnp. According to the result of
	Halperin and Krauthgamer~\cite{DBLP:conf/stoc/HalperinK03}, the problem
	\dsnp admits no polynomial-time $\log^{2-\varepsilon} n$-approximation
	algorithm for every $\varepsilon > 0$, unless $\NP \subseteq
	\mathsf{ZTIME}(n^{\polylog(n)})$. We conclude that \ftf admits no
	polynomial-time $\log^{2-\varepsilon}(\ell/2)$-factor approximation
	algorithm under the same assumption.
\end{proof}

\section{Proofs Omitted from Section~\ref{sec:FTF:apx}} \label{app:FTF:apx}

\begin{proof}[Proof of Theorem~\ref{thm:ftf:approx}]
	Let $I$ be an instance of \ftf on a digraph $D = (V, A))$ with weight
	$c \in \Z_{\geq 0}^A$, terminals $s$ and $t$,  vulnerable arcs $M$ and
	desired flow value  $\ell$.  We consider an instance $I' = (D, c, s, t,
	\ell+1, g)$ of \mcfp, where the arc capacities $g$ are given by 
	\[
		g_e \coloneqq
		\begin{cases}
			1 &   \text{if $e \in M$, and} \\
			1+ \frac 1k   &  \text{otherwise}
		\end{cases}
	\]
	An optimal solution to $I'$ can be computed computed in polynomial-time
	by standard techniques.  We saw in the discussion at the beginning of
	Section~\ref{sec:FTF:apx} that the set of arcs of positive flow in a solution
	to $I'$ yields a feasible solution to $I$. 

	It remains to bound the approximation ratio. Let $Y^*$ be an optimal
	solution to $I$ of cost $\opt(I)$. We first show that $Y^*$ contains
	$\ell+1$ disjoint $s$-$t$ paths.
	
	\setcounter{myclaim}{0}
	\begin{myclaim}
		$Y^*$ contains an $s$-$t$ flow of value $\ell+1$ with respect to the capacities $g$.	
	\end{myclaim}
	\begin{proof}
		First observe that in any feasible solution to $I$, every $s$-$t$ cut
		contains either at least $\ell$ safe arcs or at least $\ell+1$ arcs.
		Now, an $s$-$t$ cut $Z$ in $Y^*$ having at least $\ell$ safe
		arcs satisfies $g(Z) \geq (1 + \frac{1}{\ell}) \cdot \ell =
		\ell+1$.  On the other hand, an $s$-$t$ cut $Z'$ in $Y^*$
		containing at least $\ell+1$ arcs satisfies $g(Z') \geq
		\ell+1$.  Hence, each $s$-$t$ cut in $Y^*$ has capacity at
		least $\ell+1$. By the max-flow-min-cut theorem there is an
		$s$-$t$ flow of value at least $\ell+1$.
	\end{proof}

	The theorem now follows from the next claim.

	\begin{myclaim}
		\label{lem:flow-solution-guarantee}
		An optimal solution to $I'$ has cost at most $(\ell + 1)\cdot \opt(I)$.
	\end{myclaim}
	\begin{proof}
		Let $f^* \in \mathbb{Q}^A$ be an optimal $s$-$t$ flow with
		respect to the capacities $g$.  Furthermore, let $Y$ be the set
		of arcs of positive flow, that is $Y := \{ e \in A \mid f^*_e >
		0 \}$.  Let $Y_M = Y \cap M$ be the vulnerable arcs in $Y$ and
		let $Y_S = Y \setminus Y_M$ be the safe arcs. First,
		we may assume that each arc $e \in Y$ has flow value at least
		$f^*_e \geq 1/\ell$, since each arc has capacity either $1$ or
		$1+\frac{1}{\ell}$. This is true since we could scale the arc
		capacities $g$ by a factor $\ell$, which allows us to compute
		(in polynomial time) an integral  optimal solution with respect
		to the scaled capacity function, using any augmenting paths
		algorithm for \mcfp. In addition, observe that we may pay a
		factor of at most $1 + \frac{1}{\ell}$ too much for each safe
		arc since the capacity of the safe arc is $1 +
		\frac{1}{\ell}$.  Therefore, we may bound the cost of a
		safe arc $e \in Y_S$ by $ \ell \cdot (1 +
		\frac{1}{\ell}) \cdot c_e \cdot f_e$ and the cost of each
		vulnerable arc $e \in Y_M$ by $\ell \cdot c_e \cdot f_e$, where
		$f_e$ is the flow-value of arc $e$ according to the solution
		$Y$.  Hence, we obtain
		\begin{equation*}
			\begin{aligned}
				c(Y) & = c(Y_S) + c(Y_M) \\
				& \leq \ell \cdot \left( (1 + \frac{1}{\ell}) \cdot \sum_{e \in Y_S} c_e  \cdot f^*_e  + \sum_{e \in Y_M} c_e \cdot f^*_e \right)  \\
				& \leq \ell  \cdot (1 + \frac{1}{\ell}) \cdot \left( \sum_{e \in Y_S} c_e \cdot f^*_e + \sum_{e \in Y_M} c_e \cdot f^*_e \right)  \\
				& \leq (\ell+1) \cdot \opt(I)\enspace,
			\end{aligned}
			\label{eq:dorkf:lp:bound}
		\end{equation*}
		where the first inequality follows from the two arguments above
		and the last inequality follows from Claim 1.
	\end{proof}
\end{proof}

The remainder of this section is devoted to proving
Theorem~\ref{thm:ALG:DP:optimal}. For this purpose we need another auxiliary
graph, that we use as a certificate of feasibility.  For a graph $H = (V, A^*)$
such that $X_0 \subseteq A^* \subseteq A$, we denote the corresponding residual
graph by $D_{X_0}(A^*) = (V, A')$. The arc-set $A'$ is given by $A' := \{ uv
\in A \mid uv \notin X_0 \} \cup \{ vu \in A \mid uv \in X_0 \}$.
An illustration of this graph is given in
Figure~\ref{fig:augmentation:feasibility:example}. We first show that in a
feasible solution $Y \subseteq A \setminus X_0$, each vulnerable arc in $X_0$
is contained in a strongly connected component of $D_{X_0}(X_0 \cup Y)$.

\begin{figure}[t]
	\subcaptionbox{Graph $D$ and $X_0$ consisting of two disjoint paths.\label{fig:ex1}}[0.5\textwidth]
	{
			\begin{tikzpicture}
				\tikzstyle{hvertex}=[thick,circle,inner sep=0.cm, minimum size=2.2mm, fill=white]
				\tikzstyle{overtex}=[thick,circle,inner sep=0.cm, minimum size=2.2mm, fill=white, draw=black]

				\node[overtex] (v1) {};
				\node[hvertex, below of=v1] (v2) {};      
				\node[overtex, left of=v2, label=left:$s$] (s0) {};
				\node[overtex, below of=v2] (v3) {};

				\node[overtex,right of=v1] (v21) {};
				\node[hvertex,right of=v2] (v22) {};
				\node[overtex,right of=v3] (v23) {};

				\node[overtex,right of=v21] (v31) {};
				\node[overtex,right of=v22] (v32) {};
				\node[overtex,right of=v23] (v33) {};

				\node[overtex,right of=v31] (v41) {};
				\node[overtex,right of=v32] (v42) {};
				\node[overtex,right of=v33] (v43) {};
%
				\node[overtex,right of=v42, label=right:$t$] (t0) {};

				\draw[edge, ->, black] (s0) -- (v1);
				\draw[edge, ->, black] (s0) -- (v3);
				\draw[edge, -> , red] (v1) -- (v21);
				\draw[edge, -> , red] (v3) -- (v23);
				\draw[edge, ->, black] (v21) -- (v31);
				\draw[edge, ->, black] (v23) -- (v33);
				\draw[edge, -> , red] (v31) -- (v41);
				\draw[edge, -> , red] (v33) -- (v43);
				\draw[edge, ->, black] (v41) -- (t0);   
				\draw[edge, ->, black] (v43) -- (t0);

				\draw[edge, -> , red!30!white, dashed] (v1) to [bend right=60] (v21);
				\draw[edge, -> , red!30!white, dashed] (v3) -- (v32);  
				\draw[edge, ->, black!30!white, dashed] (v32) -- (v42);  
				\draw[edge, -> , red!30!white, dashed] (v42) -- (v41); 
				\draw[edge, -> , red!30!white, dashed] (v31) -- (v32); 
				\draw[edge, -> , red!30!white, dashed] (v42) -- (v43);  

				\draw[edge, -> , red!30!white] (v1) to [bend left=45] (v41);
				\draw[edge, -> , black!30!white] (v31) -- (v42); 
				\draw[edge, ->, red!30!white] (v32) -- (v33); 
			\end{tikzpicture}
		}
		\hspace{-1em}
		\subcaptionbox{Residual graph $D_{X_0}(X_0 \cup Y)$.\label{fig:augmentation:feasibility:strong:connected}}[0.5\linewidth]
		{
			\begin{tikzpicture}[vertex/.append style={node distance=6em}]
				\tikzstyle{hvertex}=[thick,circle,inner sep=0.cm, minimum size=2.2mm, fill=white]
				\tikzstyle{overtex}=[thick,circle,inner sep=0.cm, minimum size=2.2mm, fill=white, draw=black]

				\node[overtex] (v1) {};
				\node[hvertex, below of=v1] (v2) {};      
				\node[overtex, left of=v2, label=left:$s$] (s0) {};
				\node[overtex, below of=v2] (v3) {};

				\node[overtex,right of=v1] (v21) {};
				\node[hvertex,right of=v2] (v22) {};
				\node[overtex,right of=v3] (v23) {};

				\node[overtex,right of=v21] (v31) {};
				\node[overtex,right of=v22] (v32) {};
				\node[overtex,right of=v23] (v33) {};

				\node[overtex,right of=v31] (v41) {};
				\node[overtex,right of=v32] (v42) {};
				\node[overtex,right of=v33] (v43) {};
%
				\node[overtex,right of=v42, label=right:$t$] (t0) {};

				\draw[edge, <-, black] (s0) -- (v1);
				\draw[edge, <-, black] (s0) -- (v3);
				\draw[edge, <- , red] (v1) -- (v21);
				\draw[edge, <- , red] (v3) -- (v23);
				\draw[edge, <-, black] (v21) -- (v31);
				\draw[edge, <-, black] (v23) -- (v33);
				\draw[edge, <- , red] (v31) -- (v41);
				\draw[edge, <- , red] (v33) -- (v43);
				\draw[edge, <-, black] (v41) -- (t0);   
				\draw[edge, <-, black] (v43) -- (t0);    
				\draw[edge, -> , red] (v1) to [bend right=60] (v21);
				\draw[edge, -> , red] (v3) -- (v32);  
				\draw[edge, ->, black] (v32) -- (v42);  
				\draw[edge, -> , red] (v42) -- (v41); 
				\draw[edge, -> , red] (v31) -- (v32); 
				\draw[edge, -> , red] (v42) -- (v43);
			\end{tikzpicture}
		}
		\caption{Illustration of the structure of feasible solutions to
		\onerkdpa. Unsafe arcs are red, safe arcs are black. In Fig.~\ref{fig:ex1}: edges of $X_0$ are black and red; edges of $A - X_0$ are
light gray and light red. Dashed edges belong to $Y$. 
		\label{fig:augmentation:feasibility:example}}
	\end{figure}
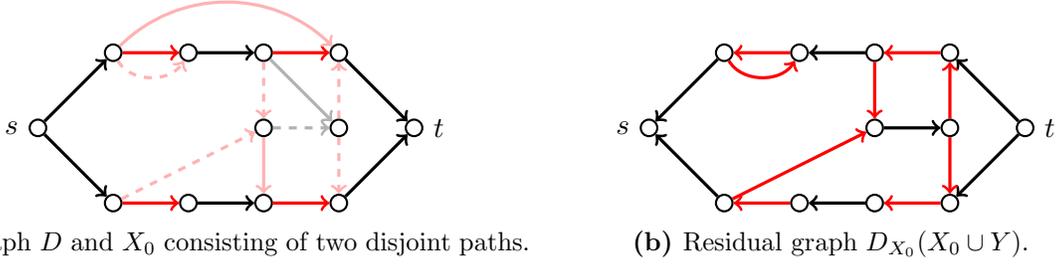

\begin{lemma}
	\label{lem:augm:residual-graph-feasibility}
	Let  $Y \subseteq A \setminus X_0$. Then $Y$ is a feasible solution to
	$I$ if and only if each vulnerable arc $f \subseteq M \cap X_0$ is
	contained in a strongly connected component of $D_{X_0}(X_0 \cup
	Y)$.
\end{lemma}
\begin{proof}[Proof of Lemma~\ref{lem:augm:residual-graph-feasibility}]
	We first prove the ``if'' part, so let $f = uv$ be a
	vulnerable arc in $X_0$ that is contained in a strongly connected
	component of $D_{X_0}(X_0 \cup Y)$.  Since $f \in X_0$, the arc $f$ is
	reversed in $D_{X_0}(X_0 \cup Y)$ and since $f$ is on a cycle $C$ in
	$D_{X_0}(X_0 \cup Y)$, there is a path $P$ from $u$ to $v$ in
	$D_{X_0}(X_0 \cup Y)$. Let $P'$ be the path corresponding to $P$ in
	$X_0 \cup Y$.  Note that $P'$ is not a directed path in $D$ and that an
	arc $e$ on $P'$ is traversed forward if $e \in P' \cap Y$ and traversed
	backward if $e \in P' \cap X_0$. We partition $P'$ into two disjoint
	parts $P'_{X_0} = P' \cap X_0$ and $P'_Y = P' \cap Y$.  We now argue
	that $(X_0 - P'_{X_0} - f) \cup P'_Y$ contains $\ell$ disjoint
	$s$-$t$ paths.
	Clearly, we have $(X_0 - P'_{X_0} - f) \cup P'_Y \subseteq X_0 \cup
	Y$. Furthermore, by our assumption that $X_0$ is the union of $\ell$
	$s$-$t$ edge-dijoint paths, for each vertex $v \in V - \{s, t\}$, we
	have $\delta^+(v) = \delta^-(v)$ and $\delta^+(s) = \delta^-(t) =
	\ell$.  Since $C$ is a cycle in $D_{X_0}(X_0 \cup Y)$ the degree
	constraints also hold for $(X_0 - P'_{X_0} - f) \cup P'_Y$.  Hence
	$(X_0 - P'_{X_0} - f) \cup P'_Y$ is the union of $\ell$ disjoint
	$s$-$t$ paths.

	We now prove the ``only if'' part. Let $f = uv \in X_0$ be a vulnerable
	arc and suppose $f$ is not contained in a strongly connected component
	of $D_{X_0}(X_0 \cup Y)$.  Let $L \subseteq V$ be the set of vertices
	that are reachable from $u$ in $D_{X_0}(X_0 \cup Y)$ and let $R
	= V - L$.  Note that $s \in L$, since $u$ is on some $s$-$t$ path in
	$X_0$ and $t \in R$, since otherwise there is a path from $u$ to $v$ in
	$D_{X_0}(X_0 \cup Y)$ (since every arc in $X_0$ is reversed in
	$D_{X_0}(X_0 \cup Y)$). 
	Let $L' = \{ x_1, \dots, x_\ell \} \subseteq L$, $x_i \in P_i$ for 
	$1 \leq i \leq \ell$, be the vertices of $L$ that are closest to $t$ in
	$X_0$.
	We now claim that $\delta^+ (L)$ is a cut of size $\ell$ in $X_0 \cup
	Y$ containing $f$.  Since $f$ is vulnerable this contradicts the
	feasibility of $X_0 \cup Y$.  We have $f \in \delta^+ (L)$ in $X_0
	\cup Y$, since otherwise $f$ is contained in a strongly connected
	component of $D_{X_0}(X_0 \cup Y)$. By the construction of $L$, we have $Y \cap \delta^+(L) = \emptyset$.
	Since $X_0$ is the union of $\ell$ disjoint paths, the set
	$\delta^+(L)$ has size at most $\ell$, proving our claim, since this
	implies that $X_0 \cup Y$ is not feasible.
\end{proof}
\begin{proof}[Proof of Theorem~\ref{thm:ALG:DP:optimal}]
	\setcounter{myclaim}{0}
	Let $\mathcal{P}$ be a shortest path in the auxiliary graph $\mathcal{D}$
	and let $Y$ be the solution computed by Algorithm \ref{ALG:DP:augprob}.
	We first establish the feasibility of $Y$.
	\begin{myclaim}\label{clm:feasibility:algo1}
		The solution $Y$ computed by Algorithm \ref{ALG:DP:augprob} is feasible.
	\end{myclaim}
	\begin{proof}
		For a link $xy \in \mathcal{P} \cap \mathcal{A}_2$, let $Y(x, y)$ be an optimal
		solution to the instance $I(x, y)$ of \dsnp. We now argue that $X_0 \cup Y$ is
		feasible to the instance $I$ of \onerkdpa. By
		Lemma~\ref{lem:augm:residual-graph-feasibility}, it suffices to
		show that each vulnerable arc of $X_0$ is contained in some
		strongly connected component of $D_{X_0}(X_0 \cup Y)$.
		Consider a vulnerable arc $x_i y_i \in X_0$ on a path $P_i$
		contained in $X_0$. Two nodes $x$ and $y$ containing $x_i$
		and $y_i$, respectively, cannot be connected by a link in
		$\mathcal{A}_1$ of $\mathcal{D}$, since $x_i y_i$ is
		vulnerable.  Let $u = (u_1, u_2, \ldots, u_\ell)$ (resp., $v =
		(v_1, v_2, \ldots, v_\ell)$) be the node on $\mathcal{P}$ such
		that $u_i$ (resp., $v_i$) is closest to $x_i$ (resp., $y_i$) on
		the subpath from $s$ to $x_i$ (resp., $y_i$ to $t$) of $P_i$.
		These two nodes exist since $\mathcal{P}$ is a path from $(s,
		\ldots, s)$ to $(t, \ldots, t)$ in $\mathcal{D}$.  If there is
		more than one such node, let $u$ be a greatest and $v$ be a
		smallest such node with respect to the ordering $\leq$.
		By this choice of $u$ and $v$,  we have that $uv \in
		\mathcal{P} \cap \mathcal{A}_2$.  Therefore, the optimal
		solution $Y(u, v)$ to $I(u, v)$ has been added to $Y$ by
		Algorithm~\ref{ALG:DP:augprob}.  Since $Y(u, v)$ connects $x_i$
		and $y_i$ in $D_{X_0}(X_0 \cup Y)$, we have that the arc $x_i
		y_i$ is contained in a strongly connected component in the
		residual graph $D_{X_0}(X_0 \cup Y)$.  
	\end{proof}
		
	Let $Y^*$ be an optimal solution to $I$ of weight $\opt(I)$. We now
	show that $Y$ computed by Algorithm~\ref{ALG:DP:augprob} is optimal.
	Observe that  the weight of $Y$ is equal to $c'(\mathcal{P})$, so it
	suffices to show that $w(\mathcal{P}) \leq \opt(I)$. To prove the
	inequality, we first introduce a partial ordering of the strong
	components of $D_{X_0}(X_0 \cup Y^*)$. Using this ordering we can
	construct a path $\mathcal{P}'$ in $\mathcal{D}$ from $(s, \ldots, s)$
	to $(t, \ldots, t)$ of cost $w(\mathcal{P'}) = \opt(I)$. We conclude
	by observing that a shortest path $\mathcal{P}$ has cost at most
	$w(\mathcal{P}')$.

	We introduce some useful notation. Let $Z$ be a strongly connected
	component of $D_{X_0}(X_0 \cup Y^*)$ and let $L(Z) = \{ i \in \{1, 2,
	\ldots, \ell\} \mid E(P_i) \cap E(Z) \neq \emptyset \}$ be the set of
	indices of the paths $P_1, ..., P_\ell$ that have at least one edge in
	common with $Z$ (ignoring orientations).  Additionally, for each $i \in
	L(Z)$, let $s_i(Z)$ be the vertex of $Z$ that is closest to $s$ on
	$P_i$ and let $S(Z) := \bigcup_{i \in L(Z)} s_i(Z)$.  Similarly, for
	each $i \in L(Z)$ let $t_i(Z)$ be the vertex of $Z$ that is closest to
	$t$ on $P_i$ and let $T(Z) := \bigcup_{i \in L(Z)} t_i(Z)$.

	\begin{myclaim}\label{clm:ordering}
		Let $e_i^1, e_i^2 \in P_i \cap M$ be two vulnerable arcs of a
		path $P_i$, such that their corresponding connected components
		$Z_1$ and $Z_2$ of $D_{X_0}(X_0 \cup Y^*)$ are disjoint.  If
		$e_i^1$ precedes $e_i^2$ on $P_i$, then $t_i(Z_1) < s_i(Z_2)$
		for every $i \in L(Z_1) \cap L(Z_2)$.
	\end{myclaim}
	\begin{proof}
		Suppose for a contradiction that there is some $i \in L(Z_1)
		\cap L(Z_2)$, such that $t_i(Z_1) \geq s_i(Z_2)$.  By the
		definition of $D_{X_0}(X_0 \cup Y^*)$, we have that $t_i(Z_1)$
		is connected to $s_i(Z_2)$ in $D_{X_0}(X_0 \cup Y^*)$.  But
		this implies that $Z_1$ and $Z_2$ are not disjoint, a
		contradiction.
	\end{proof}

	Using this claim we can construct a path $\mathcal{P'}$ in
	$\mathcal{D}$ of cost at most $\opt(I)$.

	\begin{myclaim}\label{clm:pathP'}
		There is a path $\mathcal{P'}$  from $(s, \ldots, s)$ to $(t, \ldots, t)$ in $\mathcal{D}$ of cost at most $\opt(I)$.
	\end{myclaim}
	\begin{proof}
	  We give an algorithm that constructs a path $\mathcal{P}'$ from $(s,
	  \ldots, s)$ to $(t, \ldots, t)$ in $\mathcal{D}$, such that $\mathcal{P}'$ only uses links
	  in $\mathcal{A}_2$ that correspond to strongly connected
	  components of $Y^*$ in $D_{X_0}(X_0 \cup Y^*)$. Starting from $s_1 =
	  (s, \ldots, s) \in \mathcal{V}$, we perform the  following two steps
	  alternatingly until we reach $(t, ..., t) \in \mathcal{V}$. 
	  \begin{enumerate}
	    \item From the current node $u$, we proceed by
		    greedily taking links of $\mathcal{A}_1$ until we reach a
		    node $v = (v_1, v_2, \ldots, v_\ell) \in \mathcal{V}$ 
		with the property that each vertex $v_i$, $1 \leq i \leq \ell$,
		is either $t$ or part of some strongly connected component of
		$D_{X_0}(X_0 \cup Y^*)$.\label{itm:step:1}
	    \item From the current node $v$, we take a link $vw \in
	      \mathcal{A}_2$ to some node $w \in \mathcal{V}$, where the link
	      $vw$ corresponds to a strongly connected component $Z$ of
	      $D_{X_0}(X_0 \cup Y^*)$.\label{itm:step:2}
	  \end{enumerate}
	  First, we observe that Step~\ref{itm:step:1} is well-defined: If at
	  some point we reach a node $v = (v_1, v_2, \ldots, v_\ell)$ with no
	  out-arcs in $\mathcal{A}_1$ and there is some $1 \leq i \leq \ell$,
	  such that $v_i$ is not in a strongly connected component of
	  $D_{X_0}(X_0 \cup Y^*)$, then $Y^*$ is not feasible, since $Q_i(v_i,
	  \cdot)$ contains a vulnerable arc but there is no substitute path in
	  $X_0 \cup Y^*$.  To show that we obtain an $s_1$-$t_1$ path by 
	  alternating the two steps above it remains to show
	  that whenever we are in Step~\ref{itm:step:2}, there is a link $vw
	  \in \mathcal{A}_2$ such that $vw$ corresponds to a strongly connected
	  component of $D_{X_0}(X_0 \cup Y^*)$.  Let $v = (v_1, ..., v_\ell)
	  \in \mathcal{V}$ be the current node at the beginning of
	  Step~\ref{itm:step:2}. 
	  Without loss of generality we may assume that
	  $v_i \neq t$ for $i \in \{1, 2, \ldots,
	  q\}$ for some $1 \leq q \leq \ell$ (if $v_i = t$ for all $i \in [\ell]$ then we are done). 
	  We now need to show that there is a
	  strongly connected component $Z$ in $D_{X_0}(X_0 \cup Y^*)$ such that
	  $S(Z) \subseteq \{v_1, ..., v_q\}$.
Suppose for each strongly connected component $Z$ in $D_{X_0}(X_0 \cup Y^*)$ satisfying 
$S(Z) \cap \{ v_1, ..., v_q \} \neq \emptyset$ we have that $S(Z) \nsubseteq \{ v_1, ..., v_q \}$ and
let $Z_1, ..., Z_j$ be those components.
By Claim \ref{clm:ordering} we have that for every component $Z_i, i \in [j]$ 
there is a component $Z_{i'}, i' \in [j]$,
such that there is some $l \in L_i \cap L_{i'}$ with $t_l(Z_{i'}) < s_l(Z_i)$.
But then the ordering of the sets $Z_1, ..., Z_j$ induces a cycle, a contradiction to the ordering.

	Hence the algorithm computes an $s_1$-$t_1$ path $\mathcal{P}'$;
	it remains to bound its cost. Consider any link $e = vw \in
	\mathcal{P}' \cap \mathcal{A}_2$ with $v = (v_1, ..., v_\ell)$ and $w =
	(w_1, ..., w_\ell)$ and let $Z_e$ be its corresponding strongly
	connected component in $Y^*$. 
	Let $c(Z_e) \coloneqq \sum_{a \in Z_e \setminus X_0} c(a)$, i.e.\ the cost of the edges in $Z_e$ that do
	not belong to $X_0$.
	For the cost of the link~$e$ we now have that $w_e = \opt( I(v, w) )
	\leq c(Z_e)$,
	since we compute an optimal solution to $I(v, w)$ that connects the
	terminal pairs $(v_i, w_i)_{1 \leq i \leq \ell}$ in the residual graph.
	Hence we have $w(\mathcal{P}') = \sum_{e \in \mathcal{P}' \cap
	\mathcal{A}_2} w_e \leq \sum_{e \in \mathcal{P}' \cap \mathcal{A}_2}
	c(Z_e) = \opt(I)$. 
\end{proof}

	A shortest path $\mathcal{P}$ from $s_1$ to $t_1$ in $\mathcal{D}$
	satisfies $c'(\mathcal{P}) \leq c'( \mathcal{P}') \leq \opt(I)$. 
\end{proof}

\begin{proof}[Proof of Theorem~\ref{thm:ALG:DP:runtime}]
  Let $|V| = n$ and
  $|A| = m$.  
  The auxiliary graph
  $\mathcal{D}= (\mathcal{V}, \mathcal{A})$ has order $| \mathcal{V} | \leq
  n^\ell$ and size $ | \mathcal{A} | \leq n^{2\ell}$. For each link $e \in
  \mathcal{A}$ we need to solve (at most) one instance of $\dsnp$ on at most
  $\ell$ terminal pairs, which can be done in time $O(n^{2\ell} \cdot (m
  n^{4\ell-2} + n^{2\ell-1} \log n))$ using the algorithm
  from~\cite{feldman_ruhl_06} for finding a cost-minimal strongly connected
  subgraph on $\ell$ terminal pairs. 
  Hence, $\mathcal{D}$ can be constructed in $O(mn^{6\ell-2} + n^{6\ell-1} \log
  n)$.  Since $\mathcal{D}$ is acyclic, a shortest path from $(s, \ldots, s)$
  to $(t, \ldots, t)$ in $\mathcal{D}$ can be computed in time $O(|
  \mathcal{V}| + | \mathcal{A} |) = O(n^{2\ell})$. Therefore,
  Algorithm~\ref{ALG:DP:augprob} runs in time $O(mn^{6\ell-2} + n^{6\ell-1} \log
  n)$.
\end{proof}
\section{Proofs Omitted from Section~\ref{sec:FTF:other-problems}}

\begin{proof}[Proof of Proposition~\ref{prop:dorkf-twokdst}]
	Consider an instance $I$ of \twokdst on a graph $D=(V, A)$ with edge
	weights $c \in \mathbb{Q}^A$, root $s \in V$, and $k$ terminals $T=
	\{t_1, \ldots, t_k\}$.
	We construct (in polynomial time) an instance $I' = ((V', A \cup A_t), c', s, t, A)$ of \dorkkf as follows.
	We add to $D$ a vertex $t$ and an arc from each terminal to $t$; that is, we
	add the edge-set $A_t = \{( t_i, t ) : t_i \in T\}$. Let $D'$ be the resulting graph.
	The cost function $c'$ is given by
	\[
		c'_e \coloneqq
		\begin{cases}
			c_e	&	\text{if $e \in A$, and}\\
			0 	&	\text{otherwise}.
		\end{cases}
	\] 
	Finally, we set $M \coloneqq A$, that is, all original arcs are unsafe,
	while the new arcs $A_t$ are safe.

	Let $X$ be a feasible solution to $I'$. We have $A_t \subseteq X$,
	since otherwise $X$ is not feasible. We now show that there are two
	disjoint $s$-$t_i$ paths in $(V, X \setminus A_t)$ for every $t_i \in
	T$. Assume this is not true and there is some $1 \leq i \leq k$, such
	that there are no two disjoint $s$-$t_i$ paths in $(V, X)$.  By the
	max-flow-min-cut theorem there is a cut edge $e \in A$ (or no edge at
	all) that separates $s$ and $t_i$ in $(V, X)$.  But then $Y \coloneqq
	\{(A_t \cup \{ e \}) \setminus ( t_i, t ) \}$ is an $s$-$T$ cut in $(V,
	X)$ of size $k$ containing vulnerable edge, which contradicts the
	assumption that~$X$ is feasible for $I'$.  Furthermore, every $s$-$T$
	cut in $(V, X \setminus A_t)$ contains at least $k+1$ edges since all
	edges in $A$ are unsafe, since otherwise $X$ is not feasible for $I'$.
	Finally, observe that there is a one-to-one correspondence between
	feasible solutions to $I'$ and those of the \twokdst instance $I$ with
	the additional property that every $s$-$T$ cut contains at least $k+1$
	edges.
\end{proof}

\begin{proof}[Proof of Proposition~\ref{prop:dortf-onetwotwodst}]
	Let $I$ be an instance of \onetwotwodst on a graph $D=(V, A)$ with edge-weights
	$c \in \mathbb{Q}^E$, root $s \in V$, and terminals $T= \{t_1, t_2 \}$.
	Similar to the proof of Proposition~\ref{prop:dorkf-twokdst} we construct an
	instance $I'$ of \dortf as follows. We add to $D$ two vertices $u$ and $t$ and
	four directed edges $\hat{A} = \{ (s, u), (t_1, u), (u, t), (t_2, t) \}$.
	Let the resulting graph be $D'$. The edge weights $c'$ of $D'$ are given by
	\[
		c'_e \coloneqq
		\begin{cases}
			c_e	&	\text{if  $e \in A(G)$, and}\\
			0	&	\text{otherwise}.
		\end{cases}
	\] 
	Finally, we set $M \coloneqq A \cup \{ (s,u), (t_1, u) \} $, that is, the edges incident to $t$ are safe while
	all other edges are unsafe. 

	Let $X$ be a feasible solution to $I'$. We have $\hat{A} \subseteq X$, since
	otherwise $X$ is not feasible. We now show that there is at least one
	$s$-$t_i$ path and there are at least two disjoint $s$-$t_2$ paths in $(V, X
	\setminus \hat{A})$. Assume first that there is no path from $s$ to $t_1$ in
	$(V, X \setminus \hat{A})$. But then $\{ (s, u), (v, t) \}$ is a cut of size
	two in $D'$, where $(s, u)$ is a vulnerable edge. This contradicts the
	feasibility of $X$. Now assume that there are no two disjoint $s$-$t_2$ paths
	in $(V, X \setminus \hat{A})$.  Similar to the proof of
	Proposition~\ref{prop:dorkf-twokdst} we then have a contradiction to the
	feasibility of $X$.  Finally, observe that there is a one-to-one correspondence
	between feasible solutions to $I$ and feasible solutions to $I'$. 
\end{proof}
\end{document}